\documentclass{article}
\usepackage{amsmath,a4wide}
\usepackage{amssymb}
\usepackage{amsthm,amsfonts}
\usepackage{eucal}
 \usepackage{paralist}
  \usepackage{graphics} 
  \usepackage{epsfig}   
\newtheorem{theorem}{Theorem}[section]
\newtheorem{corollary}[theorem]{Corollary}

\newtheorem{lemma}[theorem]{Lemma}

\theoremstyle{definition}

\newtheorem{remark}{Remark}[section]

\newcommand{\R}{{\mathord{\mathbb R}}}

\let\o=\omega
\let\e=\varepsilon
\let\a=\alpha
\let\pa=\parallel

\let\o=\omega
\let\e=\varepsilon
\let\a=\alpha

\title{Ghost effect by curvature in\\ planar Couette flow}

\author{Leif Arkeryd, Raffaele Esposito,  Rossana Marra  and Anne Nouri}
\begin{document}
\date{}
\maketitle
\centerline{\scshape Leif  Arkeryd}
\medskip
{\footnotesize
 \centerline{Chalmers, 41296 Gothenburg, Sweden}}
\medskip
\centerline{\scshape Raffaele  Esposito }
\medskip
{\footnotesize
 \centerline{MeMOCS, Universit\`a dell'Aquila}
   \centerline{Cisterna di Latina, 04012 LT, Italy }
 }
 \medskip
\centerline{\scshape Rossana  Marra}
\medskip
{\footnotesize
 \centerline{ Dipartimento di Fisica and Unit\`a INFN, Universit\`a di
Roma Tor Vergata}
 \centerline{00133 Roma, Italy}
}
\medskip
\centerline{\scshape Anne  Nouri}
\medskip
{\footnotesize
 \centerline{LATP,  Universit\'e d'Aix-Marseille I, Marseille, France}
}
\bigskip
\begin{abstract}
We study a rarefied gas, described by the Boltzmann equation, between two coaxial rotating cylinders in the small Knudsen number regime. When the radius of the inner cylinder is suitably sent to infinity, the limiting evolution is expected to converge to a modified Couette flow which keeps memory of the vanishing curvature of the cylinders ({\it ghost effect} \cite{So1}). In the $1$-d stationary case we prove the existence of a positive isolated $L_2$-solution to the Boltzmann equation and its convergence.
This is obtained by means of a truncated bulk-boundary layer expansion which requires the study of a new Milne problem,  and an estimate of the remainder based on a generalized spectral inequality.
\end{abstract}

\section{Introduction}
\setcounter{equation}{0}
\def\theequation{1.\arabic{equation}}
This is a revised version of the original paper \cite{AEMN0}, where a term was missing in the spectral inequality.
It is well known (see for example \cite{C,So1,EP} and references quoted therein) that the asymptotic behavior of the solutions to the Boltzmann equation, in the limit of small Knudsen numbers, is well approximated by the compressible Euler equation for a perfect gas, while the viscosity and heat conducting effects are seen as first order corrections in the Knudsen number. To get finite size viscosity effect is more delicate because of the von Karman relation \cite{vK} between the Reynolds, Knudsen and Mach numbers, $Re$, $Kn$, $Ma$:
$$Ma\sim Re Kn.$$
When $Kn$ is small, either $Re^{-1}$ or $Ma$ have to be small. Therefore, finite viscosity effects are attained only if one assumes that the Mach number is of the same order as the Knudsen number.

With the extra assumption that density and temperature profiles differ from constants at most for terms of the order of the Knudsen number, it is then possible to show that the asymptotic behavior of the solutions to the Boltzmann equation is well approximated by the incompressible Navier-Stokes equations (INS), in the sense that the average velocity field, rescaled by the Mach number, converges to a solution $u$ to INS. Moreover, the first order correction to the temperature profile converges to a solution to the heat equation with a convective term due to the rescaled velocity field $u$. Such results have been proved in several papers. An overview is provided in \cite{EP}, to which we refer for a partial list of references on the subject. We also stress that the asymptotic behavior of the solutions to the compressible Navier-Stokes-Fourier equations, in the low Mach number limit and with the same assumption on density and temperature, is the same.

When the density and temperature do not satisfy the above mentioned assumptions, the Boltzmann equation deviates from the compressible Navier-Stokes-Fourier equations. Such a discrepancy, called {\it ghost effect} \cite{So1}, is the issue we want to address in this paper.

The name is suggested by the fact that a small velocity field produces finite size modifications of the usual heat equation. These modifications are confirmed many by numerical experiments. There has been a big theoretical, numerical and experimental work on this  and for the details we refer to \cite{So1, So2}. We remark that the time dependent equations on the torus, in the diffusive space-time scaling were written in \cite{DEL} and analyzed in \cite{B}, where it is observed that similar equations had been previously considered in some special cases \cite{Ko}. Very little is known from the mathematical point of view, and the only rigorous result we are aware of is \cite{Br}. By using the techniques illustrated in the present paper, it is possible to deal with the time dependent problem in the torus, but that will be the subject of future works.

In this paper instead, we want to study a different type of ghost effect,  pointed out in \cite{SD,So1} as {\it ghost effect by curvature}. It consists in the following: if one looks at the Couette flow  between two coaxial rotating cylinders in the limit when the radius of the inner cylinder goes to infinity, one expects to obtain, as asymptotic behavior, the standard planar Couette flow. The analysis based on the Boltzmann equation does not confirm this. Indeed, an extra term appears in the limiting equations, which is reminiscent of the original structure of the problem. Hence, an infinitesimal curvature produces a finite size discrepancy, the ghost effect by curvature. The extra term gives rise to a bifurcation of the laminar stationary solution, which is absent in the standard Couette flow.

\vskip.2cm
To be more specific, we look at the behavior of a rarefied gas between two coaxial cylinders of radius $L$ and $L+1$, described by the Boltzmann equation in the diffusive space-time scaling.  In cylindrical coordinates $(r,\theta, z)\in (L,L+1)\times [0,2\pi)\times \mathbb{R}$ it is written as
\begin{equation}
\e\frac{\partial F}{\partial t}+\ v_r\frac{\partial F}{\partial r}+ \frac{v_\theta}{r}\frac{\partial F}{\partial \theta}+ v_z\frac{\partial F}{\partial z} +\frac{v_\theta}{r }\left(v_\theta\frac{\partial F}{\partial v_r}-v_r\frac{\partial F}{\partial v_\theta}\right)=\frac{1}{\e}Q(F,F),
\end{equation}
where the positive and normalized function $F(r,\theta,z,v_r,v_\theta, v_z,t)$ is the probability density of particles in cylindrical coordinates and we have denoted by $\e$  the Knudsen number. $(v_r,v_\theta,v_z)$ are the components of the velocity in the local basis associated to the cylindrical coordinates and $Q(f,g)$ is the Boltzmann collision integral for hard spheres:
\begin{equation}Q(f,g)(v)= \frac{1}{2}\int_{\mathbb{R}^3}d\/v_*\int_{S_2}d\/n
B(n,v-v_*)\big\{f'_*g'+f'g'_*-f_*g -g_*f\big\},\end{equation}
with $f', f'_*,f,f_*$ standing for $f(v'), f(v'_*),f(v), f(v_*)$ respectively, $S_2=\{n\in\mathbb{R}^3\,|n^2=1\}$, $B$ is the differential cross section $2B(n,V)=|V\cdot n|$ corresponding to hard spheres, and $v$, $v_*$ and $v'$,$v'_*$ are  post-collisional  and pre-collisional  velocities in an elastic collision:
\begin{equation}
v'=v-n(v-v_*)\cdot n,\quad v'_*=v_*+n(v-v_*)\cdot n.\end{equation}

The use of cylindrical coordinates produces a force-like term depending on the velocity, which will be referred below as {\it centrifugal force}.

We will look at the above equation in the planar limit, where one takes the radius of the inner cylinder $L$ to infinity. A convenient change of variables is the following:
$$r={L}+\frac{y+\pi}{2\pi};\quad  {L}\theta=-x;\quad v_r=v_y;\quad v_\theta=-v_x,$$
while the $z$ variable is unchanged. For simplicity, we assume the distribution $F$ invariant under rotations around the axis of the cylinders, so we drop the dependence on $\theta$. Moreover, we scale  $1/L$ proportionally to the inverse of the square of the Knudsen number:
\begin{equation}
\frac 1 L= \frac{\e^2}{c^2},\end{equation}
with a constant $c$, related to the curvature,  which will be specified below.
With these assumptions the equation becomes
\begin{equation} \label{basic}
\e\frac{\partial F}{\partial t}+\ v_y\frac{\partial F}{\partial y}+ v_z\frac{\partial F}{\partial z} +\frac{\e^2}{ c^2} \sigma(y){v_x}\left(v_x\frac{\partial F}{\partial v_y}-v_y\frac{\partial F}{\partial v_x}\right)=\frac{1}{\e}Q(F,F),
\end{equation}
with
\begin{equation}\label{sigma}\sigma(y)=\frac{2\pi}{2\pi+\frac{\e^2}{ c^2} (y+\pi)}\/\  .\end{equation}
The variable $y$ varies between $-\pi$, corresponding to the inner cylinder, and $\pi$ corresponding to the outer cylinder.

The boundary conditions on the two cylinders are assumed to be given by the diffuse reflection condition, meaning that the distribution of the incoming particles is Maxwellian:
\begin{eqnarray}\label{bc}& &F(-\pi,z,v,t)=\alpha_-(F)\tilde M_-, \quad v_y>0,\nonumber\\
\\
& & F(\pi,z,v,t)=\alpha_+(F)\tilde M_+, \quad v_y<0.\nonumber
\end{eqnarray}
We use the following notation: for $\rho>0$, $T>0$ and $u\in \mathbb{R}^3$
\begin{equation}\label{maxwellian} M(\rho,T, u;v)= \frac{\rho}{(2\pi  T)^{3/2}}\text{\rm e}^{-\displaystyle{\frac{|v-u|^2}{2T}}}\end{equation}
is the Maxwellian with density $\rho$,  temperature $T$ and mean velocity $u$. In this paper we consider the two cylinders at the same temperature $T=1$ and rotating with velocities $U_-$ and $U_+$. Therefore
\begin{equation} \tilde M_\pm= M(\sqrt{2\pi},1, (U_\pm,0,0);v),\end{equation}
where the density has been fixed so that the normalization condition
\begin{equation}\label{normal}\int_{v_y \gtrless 0}dv |v_y| \tilde M_\pm=1\end{equation}
is satisfied. The constants $\alpha_\pm$, depending on the outgoing flow at the boundaries, are determined by the condition of vanishing net flow in the radial direction:
\begin{equation}\label{zeroflow}\int_{\mathbb{R}^3}dv v_yF(y,z,v_x,v_y,v_z,t)=0, \quad \text{for } y=-\pi,\pi.\end{equation}
By using (\ref{bc}) and (\ref{normal}), one immediately gets
\begin{eqnarray}\label{alpha}
& & \alpha_-(F)= -\int_{v_y<0} dv v_y F(-\pi,z,v_x,v_y,v_z,t),\nonumber\\
\\
& & \alpha_+(F)= \int_{v_y>0} dv v_y F(\pi,z,v_x,v_y,v_z,t).\nonumber
\end{eqnarray}

We need to introduce a low Mach number assumption. Due to the particular geometry we consider here, we do not need that the full velocity field is small. Indeed the tangential component  can be of order $1$, while we need the radial and axial components to be of the order of the Knudsen number $\e$. We will use the notation $\hat v$ to denote the couple $(v_y,v_z)$. Correspondingly,  $\hat\nabla=(\partial_y,\partial_z)$.
The Mach number assumption is therefore:
\begin{equation}\label{mach} \hat u:=\int_{\mathbb{R}^3}dv \hat v F= \mathcal{O}(\e).\end{equation}
The tangential component of the velocity, denoted by $U$ is $\mathcal{O}(1)$ with respect to the Knudsen number $\e$. However, we will also need some smallness assumption on it. Therefore, we introduce another parameter $\delta$, measuring the size of $U$:
\begin{equation}\label{nomach} U:=\int_{\mathbb{R}^3}dv v_x F= \mathcal{O}(\delta).\end{equation}
The fact that $\delta$, although small, will be chosen much larger than $\e$, is responsible for the emergence of a ghost effect. In principle $\delta$ might be completely independent of  $\e$, but, for technical reasons, we will assume, in the estimate of the remainder, a specific relation between $\delta$ and $\epsilon$.
Therefore, from now on, we replace $U$ with $\delta U$ with $U=\mathcal{O}(1)$ both in $\e$ and $\delta$.
In order to get well defined equations when $\delta\to 0$, we will also assume the constant $c$ appearing in the definition of $\sigma$ (\ref{sigma}) of order $\delta$:
\begin{equation}\label{cdeltaC} c=\delta C\end{equation}
for some other constant $C$ also of order $1$.

As usual we will look for a solution to (\ref{basic}) in terms of a truncated expansion in $\e$. The collision term forces the lowest order of the expansion to be a local Maxwellian. In order to fulfil the low Mach number assumptions (\ref{mach}), (\ref{nomach}) and the boundary conditions (\ref{bc}), the lowest order has to be of the form
\begin{equation}\label{mdelta} M_\delta= M(1+\delta r, 1+\delta \tau, (\delta U, 0,0);v)=\frac{1+\delta r}{(2\pi(1+\delta \tau))^{3/2}}\exp\left[-\frac{\bar v^2}{2(1+\delta \tau)}\right],\end{equation}
where
\begin{equation} \bar v=v-(\delta U, 0,0)= (v_x-\delta U,\hat v).\end{equation}
Note that the functions $r$ and $\tau$, representing corrections of order $\delta$ to the density and temperature, vanish at the boundary because we are restricting ourselves to the case when the two cylinders are at the same temperature. On the other hand gradients of $U$ of order $\delta$ warm up the fluid in the bulk and may produce variations of the temperature and density of order $\delta^2$.

The solution is sought for in the form
\begin{equation}\label{truexpansion}
F=M_\delta+\Phi+\e\mathcal{R},\end{equation}
where
\begin{equation}\label{expansion}
\Phi=\sum_{n=1}^N\e^n F_n\end{equation}
for a suitable choice of $N$, and $\mathcal{R}$ is a remainder.
In the next section we will give the procedure to compute the functions $F_n$'s, which will be based on a kind of Hilbert expansion for the bulk parts $B_n$ of $F_n$ and a boundary layer expansion in order to restore the boundary conditions violated by the bulk terms.
The computation of the bulk terms requires the solution of a rather complex system of equations for the hydrodynamical fields $U$, $\hat u$, $r$, $\tau$ depending on $\delta$. They are
\begin{eqnarray}\label{full}\nonumber
& &\hat \nabla[r +\tau]+\delta\hat \nabla[r \tau] =0,\\   \nonumber
& &\hat\nabla\cdot\hat u+\delta(\partial_t r+\hat\nabla\cdot(\hat u r))=0,\\ \nonumber
& &(1+\delta r) \left(\partial_tU+\hat u\cdot \hat\nabla U\right)=\eta_0\hat\Delta U+ \hat\nabla\cdot (\eta_\delta\hat \nabla U),\\
& &(1+\delta r)\big ( \partial_t\hat u+ \hat u\cdot\hat\nabla\hat u\big) +\hat\nabla \mathcal{P}_2-\frac{1}{C^2}\rho U^2e_y= \eta_0\hat\Delta \hat u\\ \nonumber
&&+\hat\nabla\cdot\left(\eta_\delta\hat\nabla \hat u+\frac{\delta^2}{\mathcal{P}}\big[\sigma_1\hat\nabla \tau\otimes\hat\nabla \tau +\sigma_2\hat\nabla U\otimes \hat\nabla U \big]\right),\\ \nonumber
& &\frac 3 2 (1+\delta r)\partial_t\tau+\frac 5 2 (1+\delta r) \hat u\cdot\hat\nabla\tau= \kappa_0\hat\Delta\tau+\hat\nabla(\kappa_\delta\hat\nabla \tau)+\delta\eta|\hat\nabla U|^2. \nonumber
\end{eqnarray}
Here $\eta$, $\kappa$, $\sigma_1$ and $\sigma_2$ are suitable transport coefficients depending on $1+\delta \tau$. We have set $\eta=\eta_0+\eta_\delta$, $\kappa=\kappa_0+\kappa_\delta$, with $\eta_0$ and $\kappa_0$ the values corresponding to $\delta=0$ and $\eta_\delta$, $\kappa_\delta$ the differences. $e_y$ is the unit vector in the direction $y$, $\mathcal{P}$ is defined in (\ref{not1}) and $\mathcal{P}_2$ is an unknown pressure related to the almost incompressibility condition given by the second of the equations (\ref{full}).

When $\delta$ goes to $0$ the equations take a rather simpler form:
\begin{eqnarray}\label{zerocou}
& &\hat\nabla\cdot\hat u=0,\nonumber\\
& &\partial_tU+\hat u\cdot \hat\nabla U=\eta_0\hat\Delta U,\\
& &\partial_t\hat u+ \hat u\cdot\hat\nabla\hat u+\hat\nabla \mathcal{P}_2-\frac{1}{C^2}\rho U^2e_y= \eta_0\hat\Delta\hat u,\nonumber
\end{eqnarray}

\begin{eqnarray}\label{zerotemp}
& &\hat \nabla[r +\tau] =0,\nonumber\\
& &\frac 3 2 \partial_t\tau+\frac 5 2 \hat u\cdot\hat\nabla\tau= \kappa_0\hat\Delta\tau.
\end{eqnarray}

\noindent The equations have to be completed with Cauchy initial data and the time independent boundary conditions
\begin{eqnarray}
\label{bcmacro}
\hat u(\pm\pi, z)=0,\quad U(\pm \pi, z)=U_\pm,\quad
\tau(\pm\pi, z)=0.\nonumber
\end{eqnarray}

The first of the equations (\ref{zerotemp}) is a Boussinesq condition ensuring the constancy of the pressure to the first order in $\delta$, while the second is just the heat equation with a convective term. Note that the equations (\ref{zerocou}) are decoupled from the (\ref{zerotemp}) and can be solved independently.

The equations (\ref{zerocou}) are the equations for the planar Couette system with an extra term in the equation for $\hat u$, representing the curvature ghost effect. Their linear analysis \cite{So1,SD} shows the presence of a bifurcation controlled by the parameters $C$ and $U_\pm$, with a stationary laminar solution losing its stability and bifurcating into two stable non laminar solutions.

This paper is devoted to the analysis of the $1$-d stationary laminar solution to (\ref{basic}). The two dimensional case where bifurcation arises, will be presented in a forthcoming paper.

In Section 2 will be given a perturbative analysis in $\delta$ of the system (\ref{HD}) in order to control the difference between its solutions and those to the equations (\ref{A1}).

Due to the presence of the centrifugal force, the boundary layer expansion, as in \cite{EML, AEMN1, AEMN2}, has to include that force. This requires the solution of a Milne problem with a force, which has been given in \cite{CEM} in the presence of a potential force. The present force is different because it is not potential and depends on velocities. Therefore the arguments in \cite{CEM} require several modifications which are presented in Section 3. Finally, in Section 4 we estimate the remainder. The key ingredient to do this is a generalized spectral inequality for a perturbed linearized Boltzmann operator, already used in \cite{AEMN2} in the context of the Benard problem.  {In the inequality given in \cite{AEMN2} a term is missing. To take it into account we change the scaling to $\delta=\gamma\e^{\frac 2 3}$ and prove a suitably modified spectral inequality, incorporating part of the asymptotic expansion. The remaining part of the asymptotic expansion is of order $\delta^2$ and easy to handle.}
The main result of this paper is summarized in the following
\begin{theorem}\label{main1}
Assume $\delta=\gamma{\e}^{\frac{2}{3}}$. Then, for $\e$ and $\gamma$ small enough,
there exists a non-negative, isolated  $L_2$-solution to the problem
\begin{equation} \label{basic1}
 v_y\frac{\partial F}{\partial y} +\frac{\e^2}{ \delta^2C^2} \sigma(y){v_x}\left(v_x\frac{\partial F}{\partial v_y}-v_y\frac{\partial F}{\partial v_x}\right)=\frac{1}{\e}Q(F,F),
\end{equation}
with diffuse reflection boundary conditions
\begin{eqnarray}\label{bc1}& &F(\mp\pi,v)=\mp\tilde M_\mp\int_{v_y<0} dv v_y F(\mp\pi,v), \quad \pm v_y>0.\end{eqnarray}
Moreover, for $q=2$ and $q=\infty$,
\begin{equation}\label{basic2}
\parallel [F-M (1,1,(\delta U,0,0)]M^{-1}\parallel_{q,2}\le c\gamma\e^{\frac{4}{3}},
\end{equation}
where $U $ is the unique solution to (\ref{A1}) and $\parallel\ \cdot\  \parallel_{q,2}$ is defined in (\ref{norm}).
\end{theorem}
\begin{remark} The proof of the theorem shows that the rest term is of order $\e^{\frac{5}{3}}$ in $L_\infty$.\end{remark}

\bigskip

\section{Expansions}
\setcounter{equation}{0}
\setcounter{theorem}{0}
\setcounter{proposition}{0}
\setcounter{remark}{0}
\def\theequation{2.\arabic{equation}}
\def\thetheorem{2.\arabic{theorem}}
\def\theproposition{2.\arabic{proposition}}
\def\theremark{2.\arabic{remark}}

In this section we show how to compute the contributions $F_n$ for $n=1,\dots,N$ in (\ref{expansion}). A modified Hilbert expansion is used to compute the bulk terms. Since they violate the boundary conditions, we introduce boundary layer corrections essentially supported in thin layers (of size of the order of $\e$) near the inner and outer cylinders. Therefore, $F_n$ is written as follows:
\begin{equation}\label{Fn} F_n=B_n+b_n^++b_n^-,\end{equation}
where $B_n$ is a smooth function of $y$, while $b_n^\pm$ are smooth  exponentially fast decaying functions of the rescaled variables $Y^\pm=\e^{-1}(\pi\mp y)$, so that they are exponentially small away from $\pm \pi$.
\subsection{The bulk expansion}

In order to compute the expression of the $B_n$, we substitute (\ref{truexpansion}) in (\ref{basic1}). We ignore the terms $b_n^\pm$, because they are assumed exponentially small, and equate terms with the same power of $\e$ up to the order $N$. We use the short notation
\begin{eqnarray}
&& \mathcal{L}_\delta f= 2 Q(M_\delta, f);\\
&& \mathcal{N}(f)={v_x}\left(v_x\frac{\partial f}{\partial v_y}-v_y\frac{\partial f}{\partial v_x}\right).\end{eqnarray}
Moreover, since in this subsection the parameter $\delta$ is kept fixed, we omit the index $\delta$ when there is no ambiguity.
We get the following conditions:
\begin{eqnarray}
& &\mathcal{L}\/B_1=v_y\partial_y M,\label{condition1}\\
& & \mathcal{L}\/B_2=v_y\partial_y B_1-Q(B_1,B_1),\label{condition2}\\
& & \mathcal{L}\/B_3=v_y\partial_yB_2+\frac{1}{C^2\delta^2}\mathcal{N}(M) - 2Q(B_1,B_2)\label{condition3},\\
\mathcal{L}\/B_n&=&v_y\partial_y B_{n-1}+\frac{1}{C^2\delta^2}\sum_{h=0}^{n-3}\sigma^{(h)}\mathcal{N}(B_{n-3-h})\nonumber\\&-&\sum_{h,k\ge1, h+k=n}Q(B_h,B_k), \quad n=4,\dots, N.\label{conditionn}
\end{eqnarray}
In (\ref{conditionn}) $B_0\equiv M$ and $\sigma^{(h)}$ are the coefficients of the $\e$-power series expansion of $\sigma(y)$:
\[\sigma(y)=\sum_{h=0}^\infty \e^h \sigma^{(h)}(y).\]
The appropriate functional space to solve the above equations is  the Hilbert space $\mathcal{H}$ of the real measurable functions on the velocity space $\mathbb{R}^3$, equipped with the inner product
\begin{equation}\label{innerprod} (f,g)=\int_{\mathbb{R}^3} dv M^{-1}(v) f(v)g(v).\end{equation}
The operator $\mathcal{L}$ is defined in a suitable dense submanifold $\mathcal{D}_\mathcal{L}$ of $\mathcal{H}$ and satisfies the following properties:
\begin{enumerate}
\item[L1)] $\mathcal{L}$ is symmetric and non positive: $(f,\mathcal{L}g)=(g,\mathcal{L}f)$; $(f,\mathcal{L}f)\le 0$.
\item[L2)] $\mathcal{L}$ has a $5$-dimensional null space spanned by the collision invariants:
\begin{equation}\label{null}\text{\rm Null\,} \mathcal{L}=\text{span}\{\psi_0,\dots,\psi_4\},\end{equation}
with $\psi_\beta=\chi_\beta M$, $\beta=0,\dots,4$ and
\begin{equation}\label{collinv}\chi_0=1;\quad \chi_1=v_x;\quad \chi_2=v_y;\quad\chi_3=v_z;\quad \chi_4= \frac{|v|^2}2.\end{equation}
The orthogonal projector on $\text{\rm Null\,} \mathcal{L}$ is denoted $P$, while $P^\perp=1-P$ denotes the projector on the orthogonal complement of $\text{\rm Null\,} \mathcal{L}$ in $\mathcal{H}$.
\item[L3)] The range of $\mathcal{L}$ is orthogonal to $\text{\rm Null\,} \mathcal{L}$: $(\psi_\alpha, \mathcal{L}f)=0$ for any $\alpha=0,\dots,4$ and for any $f\in \mathcal{D}_\mathcal{L}$.
\item[L4)] The following decomposition holds:
\begin{equation}\label{nu+K} \mathcal{L}f=-\nu f+\mathcal{K} f\end{equation}
where $\mathcal{K}$ is a compact operator and $\nu$ a smooth function such that
\begin{equation}\label{stimanu}\nu_0(1+|v|)\le \nu(v)\le \nu_1(1+|v|) \text{ for all } v\in \mathbb{R}^3.\end{equation}
\item[L5)] If $g\in P^\perp\mathcal{H}$ then, by the Fredholm alternative theorem and L4), there is a unique solution in $P^\perp\mathcal{H}$ to the equation
\begin{equation}\label{Lf=g} \mathcal{L}f =g,\end{equation}
which, with a slight abuse of notation, we denote by $\mathcal{L}^{-1}g$:
\begin{equation} f=\mathcal{L}^{-1} g.\end{equation}
Any solution in $\mathcal{H}$ of (\ref{Lf=g}) can be written as
\begin{equation}\label{general} f=\mathcal{L}^{-1}g +\bar f\end{equation}
with $\bar f\in \text{\rm Null\,} \mathcal{L}$.
\item[L6)] Spectral inequality: there is a constant $c>0$ such that
\begin{equation}\label{spectral} (f,\mathcal{L}f)\le -c(P^\perp f,\nu P^\perp f).\end{equation}
\end{enumerate}

By L3) in order to solve (\ref{condition1}) we need to impose the compatibility condition $P(\hat v\cdot\hat \nabla M)=0$. It is immediate to check that this is true if and only if
\begin{equation}\label{pressure0} \partial_y \mathcal{P}=0, \end{equation}
where
\begin{equation}\label{not1}
\mathcal{P}=\rho T, \quad  \rho=1+\delta r, \quad T=1+\delta\tau.\end{equation}
This is just the second of (\ref{HD}).

If this condition is satisfied, then
\begin{equation}\label{tildeAB}v_y\partial_y M= \delta(\tilde{\mathfrak{A}}\partial_y U+\tilde{\mathfrak{B}}\partial_y \tau),\end{equation}
where
\begin{eqnarray}
&&\tilde{\mathfrak{B}}= \bar v_x v_yM,\nonumber\\
\\
&&\tilde{\mathfrak{A}}= \frac{\bar{v}^2- 5 T}{2T^2}v_y M\nonumber\end{eqnarray}
are in $P^\perp\mathcal{H}$. We define ${\mathfrak{A}}$ and ${\mathfrak{B}}$ as the solutions in $P^\perp\mathcal{H}$ of the equations
\begin{equation}\label{AB}
\mathcal{L}{\mathfrak{A}}=\tilde {\mathfrak{A}};\quad \mathcal{L}{\mathfrak{B}}=\tilde {\mathfrak{B}}.\end{equation}
Therefore, by L5)
\begin{equation} {\mathfrak{A}}=\mathcal{L}^{-1}\tilde {\mathfrak{A}}, \quad {\mathfrak{B}}=\mathcal{L}^{-1}\tilde {\mathfrak{B}}.\end
{equation}
Henceforth,
\begin{equation}\label{B1}
B_1=\delta ({\mathfrak{B}}\partial_y U+ {\mathfrak{A}}\partial_y \tau)+
M\left(
\frac{\rho_1}{\rho}+
\frac{\bar v\cdot u}{T}+
\frac{\bar{v}^2 -3 T}{2 T^2}\tau_1
\right):= B_1^\perp+B_1^\parallel,
\end{equation}
with $\rho_1$, $\tau_1$ and $u$ to be determined by the other conditions.

In order to solve (\ref{condition2}), we need to impose the orthogonality of the right hand side of (\ref{condition3}) to the null space of $\mathcal{L}$.

A standard computation shows that this is equivalent to the conditions
\begin{eqnarray}&& \partial_y(\rho u_y)=0\nonumber\\
&&\rho \left(u_y\partial_y U\right)=\partial_y (\eta\partial_y U),\\
&&\partial_y(T\rho_1+ \rho\tau_1)=0,\nonumber\\
&&\frac 5 2 \rho u_y\partial_y \tau = \partial_y(\kappa\partial_y \tau)+\delta \eta(\partial_y U)^2,\nonumber
\end{eqnarray}
with
\[\eta=-\int_{\mathbb{R}^3} dv {\mathfrak{B}} \tilde {\mathfrak{B}},\quad \kappa=-\int_{\mathbb{R}^3} dv {\mathfrak{A}}\tilde {\mathfrak{A}}.\]
Then we can compute the part  $B_2^\perp$ of $B_2$ as before:
\begin{equation}B_2= -\mathcal{L} ^{-1}Q(B_1,B_1)+\mathcal{L}^{-1}P^\perp(\hat v\cdot \hat\nabla B_1)+B_2^\parallel,\end{equation}
with
\begin{equation}B_2^\parallel=M\left(
\frac{\rho_2}{\rho}+
\frac{\bar v\cdot u_2}{T}+
\frac{\bar{v}^2 -3 T}{2 T^2}\tau_2
\right),\end{equation}
and $\rho_2$, $u_2$ and $\tau_2$ to be determined.

The same procedure is applied to solve the equation for $B_3$, (\ref{condition3}). The compatibility condition is
\begin{equation}P\Big(v_y\partial_yB_2+ \frac{1}{C^2\delta^2}\mathcal{N}(M)\Big)=0.\end{equation}
This implies several conditions. We write explicitly only the following:
\begin{equation}\label{momdelta}
\rho u_y\partial_y u_y +\partial_y \mathcal{P}_2-\frac{1}{C^2}\rho U^2=\partial_y\left(\eta\partial_y u_y+\frac{\delta^2}{\mathcal{P}}\big[\sigma_1(\partial_y \tau)^2 +\sigma_2(\partial_yU)^2 \big]\right).\end{equation}
Here $\sigma_i$ are some suitable transport coefficients of higher order whose explicit expression is given for example in \cite{So1}.

Note that the term in $U^2$ in equation (\ref{momdelta}) derives from the contribution \newline $\int dv v_y \mathcal{N}(M)=\mathcal{O}(\delta^2)$. The result in (\ref{momdelta}) is independent of $\delta$ because of the scaling (\ref{cdeltaC}) and hence it persists in the limit $\delta\to 0$.

The equations written so far are just the system (\ref{HD}). They represent a system in the unknown functions $U$, $u_y$, $\tau$ and $r$, which does not include any of the extra funtions $\rho_1$, $\tau_1$ etc, which also have to satisfy some extra conditions, for example the Boussinesq condition to the first order in $\e$.

We do not give the explicit conditions which follow in a rather standard way. It is clear that the above procedure can be continued to any specific order. In this paper it will be truncated at  $N=5$. Note that at each step the solution is given up to the choice of five arbitrary functions which are fixed in the subsequent steps. In particular the last term of the truncated expansion, $B_N$, is determined up to its hydrodynamic part which is arbitrary. We will take advantage of this arbitrariness when dealing with the equation for the remainder.

\subsection{The boundary layer expansion}

We need to include in our scheme a boundary layer expansion, because the bulk terms $B_1$ do not satisfy the diffuse reflection boundary condition. For example, it is immediate to check from equation (\ref{B1}) that $B_1$ cannot be proportional to the Maxwellian $\tilde M_\pm$ at the boundaries. Therefore we introduce the corrective terms $b_1^\pm$, with a fast dependence on $y$, so that they are sensibly different from $0$ only close to the boundary.  To achieve this, $b_1^\pm$ is assumed to be a smooth function of the variable $Y^\pm=\e^{-1}(\pi\mp y)$. We define $\bar b_1^\pm$ as the solution to the following equation:
\begin{equation}\label{b1}
v_y\frac{\partial \bar b_1^\pm}{\partial Y^\pm} + \frac{\e^3}{C^2\delta^2}\tilde\sigma(\mp(\e Y^\pm-\pi))\mathcal{N}(\bar b_1^\pm)
=\mathcal{L}^\pm \bar b_1^\pm+\tilde{\mathcal{L}}^{\pm}_\vartheta  \bar b_1^\pm ,\end{equation}
with $\mathcal{L}^\pm g=2Q(M_\pm, g)$, $M_\pm=M(1,1,(\delta U_\pm, 0, 0);v)$.
Here $\tilde\sigma=\sigma\varphi$, with $\varphi$ a smooth cutoff function
\[\varphi(y)=\begin{cases} 1\quad y\in [0,\zeta],\\0\quad y>2\zeta\end{cases}
\]
for some $\zeta>0$ and with uniformly bounded derivatives. The operator $\tilde{\mathcal{L}}^{\pm}_\vartheta$ is defined in the same way as $\mathcal{L}^\pm$, but we replace the hard spheres collision cross section  $B(n,V)$ with $\vartheta B(n,V)^2$. The reason for introducing this unphysical operator is due to a technical difficulty which will be discussed in the next section. The parameter $\vartheta$ is chosen as $\vartheta= \frac{\e^3}{C^2\delta^2}$, and the contribution from $\tilde{\mathcal{L}}^{\pm}_\vartheta$, which should not be there, will be subtracted in the next order of the boundary layer expansion.
We prescribe vanishing mass flux at the boundary:
\[ m_1^\pm=\int_{\mathbb{R}^3} dv v_y \bar b_1^\pm(0,v)=0.\]
This equation has to be solved with prescribed incoming data at $Y^\pm=0$:
\[ \bar b_1^\pm(0, v)=h_1^\pm(v), \quad \text{ for } v_y>0.\]
The incoming boundary data are chosen in such a way to compensate the fact that $B_1^\perp$ is not proportional to a Maxwellian at the boundary: $h^\pm(v)=-B_1^\perp(\pm \pi)$. The solution to the Milne problem in general has a finite, but not vanishing limit at infinity, achieved exponentially fast. Let it be denoted by $\bar b^\pm_{1,\infty}$.
It belongs to the null space of $\mathcal{L}^\pm$.
The non vanishing of $\bar b^\pm_{1,\infty}$ is not good to our purposes because this contributes to the solution in the bulk.
Therefore we define $b_1^\pm=\bar b^\pm_{1}-\bar b^\pm_{1,\infty}$.
In this way we ensure the decay at infinity, but $b_1^\pm$ do not satisfy any more the equation (\ref{b1}) , because a term of the form $ \frac{\e^3}{C^2\delta^2}\tilde\sigma(\mp(\e Y^\pm-\pi))\mathcal{N}(\bar b_{1,\infty}^\pm)$ appears in the right hand side. We will compensate it with a term in the next order of the boundary layer expansion. We are not yet done, because the boundary value of $F_1^\pm$ would still be incorrect for the term $-\bar b^\pm_{1,\infty}$ which is not Maxwellian. However, we have not yet fixed the boundary values of $B_1^\parallel$ and we can use them  to compensate it, since it is in the null space of $\mathcal{L}^\pm$.
Finally, note that, on each boundary, the boundary value correction due to the other boundary is not zero, but exponentially small in $\e^{-1}$. This will be compensated in the remainder.
In conclusion
$f_1=B_1+b_1^++b_1^-$ satisfies the diffuse reflection boundary conditions up to terms $\Psi_1^\pm$ exponentially small in $\e^{-1}$.

The equation (\ref{b1})  is a special case of the Milne problem we discuss in the next section. We note however that in the standard Milne problem  the second term in the left hand side of (\ref{b1}) is absent. When in the equation there is a force term, as in the present case and in \cite{EML, AEMN1, AEMN2}, although very small, the lack of regularity in the velocity of the solution to the Milne problem for $v_y=0$ (the derivative $\partial_{v_y} b_1^\pm$ does not exist for $v_y=0$ at the boundary),  does not allow us to include them in higher order terms of the expansion. Indeed we need to keep it in (\ref{b1}) which will be solved in a suitably weak sense, because we cannot afford to have any $v_y$ derivative of $b_1^\pm$ present in the expansion. This problem was already present in the case of the Benard problem where the force derives from a potential and the solution is given in \cite{CEM}.

The corrections to $B_n$, for $n>1$  will be $b_n^\pm$ solving a similar equation:
\begin{equation}\label{bn}
v_y\frac{\partial \bar b_n^\pm}{\partial Y^\pm} + \frac{\e^3}{C^2\delta^2}\bar\sigma(\mp(\e Y^\pm-\pi))\mathcal{N}(\bar b_n^\pm)
=\mathcal{L}^\pm \bar b_n^\pm+\tilde{\mathcal{L}}^{\pm}_\vartheta  \bar b_n^\pm+ S_n^\pm,\end{equation}
with prescribed incoming data at $Y^\pm=0$:
\[ \bar b_n^\pm(0, v)=h_n^\pm(v), \quad \text{ for } v_y>0,\]
vanishing mass flux at the boundary:
\[ m_n^\pm=\int_{\mathbb{R}^3} dv v_y \bar b_n^\pm(0,v)=0\]
and  then define $b_n^\pm= \bar b_n^\pm - \bar b_{n,\infty}^\pm$.
The incoming boundary data are chosen in such a way that $f_n=B_n+b_n^++b_n^-$ satisfies the diffuse reflection boundary conditions up to terms exponentially small in $\e^{-1}$, $\Psi_n^\pm$.
The source term $S_n^\pm$ has the following form, for $n>1$:
\begin{eqnarray}\label{source}
&&S_n^\pm=-\tilde{\mathcal{L}}^{\pm}_\vartheta  \bar b_{n-1}^\pm+ \sum_{h,k\ge 1, h+k=n}\Big[2Q(B_h, b_k^\pm) +Q(b_h^\pm, b_k^\pm)+Q(b_h^\pm,b_k^\mp)
\nonumber\\
&&- \frac{\e^2}{C^2\delta^2}\tilde\sigma^c(\mp(\e Y^\pm-\pi))\mathcal{N}(b_{n-1}^\pm)-
\frac{\e^2}{C^2\delta^2}\tilde\sigma(\mp(\e Y^\pm-\pi))\mathcal{N}(\bar b_{n-1,\infty}^\pm)\Big]\nonumber
\end{eqnarray}
with $\tilde\sigma^c=\sigma(1-\varphi)$ and $b_0^\pm=0$ and the property $\int dv S_n^\pm(y,v) =0$.
The existence of the solutions to (\ref{b1}), (\ref{bn}) with the prescribed conditions follows from  Theorem \ref{milneth}, given in next section, via a procedure which is the same used in \cite{ELM2, ELM3, AEMN1, AEMN2}. We do not repeat it here and refer to those papers for details.

\subsection{Hydrodynamical expansion}

  In this section we compare the solution of  the stationary $1$-d equations for  $\delta>0$ with the solution of the limiting equations for $\delta=0$. The former are
\begin{eqnarray} \label{HD}
&&\partial_y(\rho u_y) =0,\nonumber\\
&&\partial_y[(1+\delta r)(1+\delta\tau)]=0,\nonumber\\
&&\rho  u_y \partial_y U=\partial_y (\eta\partial_y U),\nonumber\\
&&\frac 5 2 \rho u_y \partial_y\tau = \partial_y(\kappa\partial_y \tau)+\delta(\eta\partial_y U)^2,\\
 &&\rho u_y\partial_y u_y+\partial_y \mathcal{P}_2-\frac{1}{C^2}\rho U^2= \eta_0\partial^2_y u_y\nonumber\\&& +\partial_y\left(\eta_\delta\partial_y u_y
+\frac{\delta^2}{\mathcal{P}}\big[\sigma_1(\partial_y \tau)^2 +\sigma_2(\partial_y U)^2 \big]\right),\nonumber
\end{eqnarray}
with $\rho=1+\delta r$ and
 the boundary conditions
$$u_y(\pm \pi)=0; \quad U_x(\pm\pi)=U_\pm,\quad \tau(\pm \pi)=r(\pm \pi)=0.$$
In the limit of vanishing $\delta$ the velocity $u_y$ and $\tau$ are identically zero and \begin{eqnarray} \label{A1}
&&\eta_0\partial_y^2 U=0,\\
 && \partial_y \mathcal{P}_2-\frac{1}{C^2}\rho U^2= 0 \nonumber\end{eqnarray}
whose solution is the laminar field $\bar U=U_-+\beta(y+\pi)$, with $\beta=(2\pi)^{-1}(U_+-U_-)$.

The first equation in (\ref{HD}) implies, by using the boundary conditions for $u_y$, $\rho u_y=0$. Since $\rho=1+\delta r$ for $\delta$ small is strictly larger than zero, it has to be $u_y=0$. The equations reduce to
\begin{eqnarray}
&&\partial_y (\eta\partial_y U)=0,\nonumber\\
&&\partial_y(\kappa\partial_y \tau)+\delta(\eta\partial_y U)^2=0,\\
&&\partial_y \mathcal{P}_2-\frac{1}{C^2}\rho U^2=\frac{\delta}{\mathcal{P}}\partial_y(\sigma_1(\partial_y \tau)^2).\nonumber
\end{eqnarray}
The first two equations decouple from the third and can be solved to find $U$ and $\tau$. Then the last one gives $ \mathcal{P}_2$.

We define  $\tilde U=U-\bar U$. We have
\begin{eqnarray}\label{AA}
&&\partial_y (\eta\partial_y \tilde U)+ \beta\partial_y \eta=0,\nonumber\\
&&\partial_y(\kappa\partial_y \tau)+\delta\eta(\beta +\partial_y \tilde U)^2=0,\\
&& \tilde U(\pm\pi)=0;\quad \tau(\pm\pi)=0.\nonumber
\end{eqnarray}

The functions $\eta$ and $\kappa$ are smooth functions of the temperature. The solutions are constructed by an iterative procedure. We therefore assume that $\|\tau\|_\infty<1$ . In consequence $\|\eta\|_\infty$ and $\|\kappa\|_\infty$ are uniformly bounded for $\delta<1$. Moreover, for $\delta$ sufficiently small,  $\|\eta\|_\infty>\frac{\eta_0}2$ and $\|\kappa\|_\infty>\frac{\kappa_0}2$
with $\eta_0$ and $\kappa_0$ the values of $\eta$ and $\kappa$ for $\tau=0$.  Finally, note that $\partial_y \eta=\delta\eta'\partial_y\tau$ with $\eta'=\frac{\partial \eta}{\partial T}$  uniformly bounded.
By multiplying the first of (\ref{AA}) by $\tilde U$ and the second by $\tau$ and integrating in $y$, after an integration by parts we get:

\begin{eqnarray*}\label{hydrest}
&&\frac {\eta_0}2\|\partial_y \tilde  U\|^2\le  c\delta\|\tilde U\|\|\partial_y \tau\|\le c\delta\|\partial\tilde U\|\|\partial_y \tau\|,\\
&&\frac {\kappa_0}2\|\partial_y \tau\|^2\le c\delta\|\tau\|_\infty(\beta^2+\|\partial \tilde U\|^2)
\le c\delta\|\partial\tau\|(\beta^2+\|\partial \tilde U\|^2).
\end{eqnarray*}

Therefore
\begin{eqnarray*}\label{hydrest2}
&&\|\partial_y \tilde  U\|\le c\delta\|\partial_y \tau\|,\\
&&\|\partial_y \tau\|\le c\delta(\beta^2+\|\partial_y \tilde U\|^2).
\end{eqnarray*}
The above inequalities imply that $\|\partial_y\tilde U\|$ and $\|\partial_y\tau\|$ are $\mathcal{O}(\delta)$. In particular $\|\partial_y\tau\|_\infty<1$ for $\delta$ sufficiently small. We omit the proof of the convergence of the approximating sequence, which follows along the same lines.

We conclude with the following
\begin{theorem} If $\delta$ is sufficiently small, the equations (\ref{HD}) have a unique  $C^\infty(-\pi,\pi)$ stationary solution which differs from the laminar solution (namely  $U=\bar U(y)$, $\tau=0$) for $\mathcal{O}(\delta)$.
\end{theorem}
Call $(\rho^\delta,  U^\delta,u_y^\delta, \tau^\delta)$ the solution of (\ref{HD}) (remember $u_y^\delta=0$) and $(1, U, 0,1)$ the   solution of the  equations for $\delta=0$ with $U$ the unique solution of (\ref{A1}).\newline  Let $M_\delta=M(\rho^\delta, (\delta U^\delta,0), \tau^\delta)$ and  $M_\delta^0=M(1, ( \delta U, 0),1)$. Then
\begin{corollary}\label{cor}
For $q=2,\infty$
$$
\parallel M^{-1}[M_\delta-M_\delta^0]\parallel_{q,2}\le c \delta^2,\quad  \parallel M^{-1}[M_\delta^0-M]\parallel_{q,2}\le c\delta.$$\end{corollary}
\bigskip

\subsection{Estimates for the expansion}
The properties of $F_n$'s constructed in this section are summarized in the following theorem:
\begin{theorem}
\label{fj0}
The functions $F_n$, $n=1,\dots,5$ and $\Psi_{n,\e}$ can be determined so as to satisfy the boundary  conditions
\begin{eqnarray*}
&&F_n(\mp\pi,v)={\tilde{M}_\mp(v)}\int_{w_y \lessgtr 0} |w_y|[F _n(\mp\pi,w)-\Psi_{n,\e}(\mp\pi,w)]dw\\&&
\hskip 2cm+\Psi_{n,\e}(\mp\pi,v),
 \hspace*{0.05in}v_z\gtrless0,\nonumber
\end{eqnarray*}
and the normalization condition
 $
\int_{\mathbb{R}^3\times[-\pi,\pi]}dv\/dy
F_n=0,$
so that the asymptotic expansion in $\e$ for the stationary problem (\ref{basic1}), truncated to the order $5$ is given by
$$\Phi=\sum_{n=1}^5\e^n{ F_n}(y,v).$$
The functions $F_n$'s satisfy the conditions  {(here $\zeta_j=(1+|v|)^j$)}
$$\parallel  {\zeta_j} M^{-1}F_n\parallel _{2,2}<\infty,
\quad \parallel  {\zeta_j} M^{-1} F_n\parallel _{\infty,2}<\infty\ ,
\quad n=1,\dots,5,$$
 {for any $j$.}
\noindent The functions $\Psi_{n,\e}$ are such that
$\|\Psi_{n,\e}\|_{q,2,\sim}$, $q=2,\infty$, are exponentially small as $\e\to 0$ and $\int_{\R^3}dv v_y  \Psi_{n,\e}=0$.
\end{theorem}
\begin{proof}It follows as in \cite{AEMN1, AEMN2}, using the Theorem \ref{milneth} in next section. \end{proof}

\section{Milne problem}
\setcounter{equation}{0}
\setcounter{theorem}{0}
\setcounter{proposition}{0}
\setcounter{remark}{0}
\def\theequation{3.\arabic{equation}}
\def\thetheorem{3.\arabic{theorem}}
\def\theproposition{3.\arabic{proposition}}
\def\theremark{3.\arabic{remark}}

In this section we deal with the Milne problem
\begin{eqnarray}\label{Milne}
&&v_y\frac{\partial g}{\partial Y}+G\omega(Y){N}g={L}g+S,\nonumber\\
&&g(0,v)=h(v),\quad v_y>0,\\
&&\int dvv_y M(v)g(0,v)=0,\nonumber\\
&&\int dv S(y,v) M=0\nonumber
\end{eqnarray}
for $S$ and $h$ prescribed.
Here $M$ is the Maxwellian with $T=1$, $\rho=1$ and $u=(\mathfrak{U},0,0)$;  { in this section we adopt the notation}
\[{L}f=2M^{-1}[\mathcal{L}(Mf)+\mathcal{L}_\vartheta(Mf)], \quad Nf=M^{-1}\mathcal{N}(Mf).\]  and $\omega(Y)$ a compactly supported smooth function.
Moreover, we assume that there is $c>0$ such that
\begin{equation}\label{123}\int_{v_y>0} dvv_y M(v)h^2(v)<c.\end{equation}
Note the explicit expression of $Nf$:
\begin{equation}\label{Nf}Nf= v_x^2\frac{\partial f}{\partial v_y}-v_xv_y\frac{\partial f}{\partial v_x}-\mathfrak{U}v_xv_y f.\end{equation}
The results will be applied with $G=\frac{\e^3}{\delta^2C^2 }$, $\vartheta>G\mathfrak{U}$, $\omega=\tilde\sigma$ and $\mathfrak{U}=\delta U_\pm$.
The procedure is the same used in \cite{CEM}: we construct the solution in a slab of size $\ell$ with reflecting boundary condition $g(\ell, Rv)=g(\ell,v)$, $Rv=(v_x,-v_y, v_z)$  and obtain estimates uniform in $\ell$, then we take the limit $\ell\to \infty$. As in \cite{CEM}, the main point is to discuss the case $S=0$.  The assumptions on $S$ will be given in Theorem \ref{milneth} below. We only point out the differences with \cite{CEM}.
We write $g=q+w$ with $q\in \text{Null} \ L$ and $(q,w)=0$. We set $\bar\chi_0=1$, $\bar\chi_1=v_x-\mathfrak{U}$, $\bar\chi_2=v_y$, $\bar\chi_3=v_z$, $\bar\chi_4=\frac 1 2[(v_x-\mathfrak{U})^2+v_y^2+v_z^2]$,
and
$q=\sum_{\alpha=0}^4 b_\alpha(Y) \bar\chi_\alpha$. Moreover, $(\,\cdot\, ,\,\cdot\,)$ denotes the inner product on $L_2(\mathbb{R}^3, Mdv)$.

Note that $b_2=0$. Indeed, by multiplying (\ref{Milne}) by $M$ and integrating in $dv$, we get
\[\frac{\partial}{\partial Y}(v_y,g)+ G\omega (v_y,g)=0,\]
because
\[(1,Ng)=\int dv gMv_y.\]
Therefore, with $\Omega(Y)=\int_0^Y dY'\omega(Y')$, we have
\[b_2(Y)=(v_y,g)(Y)=\exp\{- G[\Omega(\ell)-\Omega(Y)]\}(v_y,g)(\ell)=0\]
because $(v_y,g)(\ell)$ vanishes by the reflecting boundary conditions at $Y=\ell$.

As a consequence, $q=\sum_{\alpha\neq 2} b_\alpha(Y) \bar\chi_\alpha$.
Moreover, $(v_y\bar\chi_\alpha,\bar\chi_\beta)=0$ for $\alpha,\beta\neq 2$. Therefore $(v_yq,q)=0$.

Set $I_\alpha=(v_y\bar\chi_\alpha, g)=(v_y\bar\chi_\alpha, w)$ for $\alpha\neq 2$. The functions $I_\alpha$ satisfy the following equations:
\[\frac{\partial}{\partial Y}I_\alpha=\omega\sum_{\beta\neq 2} \mathcal{S}_{\alpha,\beta} I_\beta,\]
with
\[\mathcal{S}_{\alpha,\beta}=\delta_{\alpha,\beta} + (\delta_{\alpha,1}-\mathfrak{U}\delta_{\alpha,4})\delta_{1,\beta},\]
for $\alpha,\beta\neq 2$.
Indeed,
\[(\bar\chi_\alpha, N\bar\chi_\beta)=\int dv M\bar\chi_\alpha\Big(\mathfrak{U}v_xv_y(\delta_{\beta,4}-1)-v_xv_y\delta_{\beta,1}\Big),\]
which is odd in $v_y$ for $\alpha\neq 2$. On the other hand,
\[(\bar\chi_\alpha,Nw)= \int dv \Big[v_yM\bar\chi_\alpha w+ Mv_yv_xw(\delta_{\alpha,1}-\mathfrak{U}\delta_{\alpha,4})\Big].\]
Therefore, with $I=(I_\alpha)_{\alpha\neq 2}$, we have
\[I(Y)= I(\ell)\exp\{(\Omega(\ell)-\Omega(Y))\mathcal{S}\},\]
which implies
\[I(Y)=0\]
by the reflection boundary condition at $Y=\ell$.

We take the inner product of the first equation in (\ref{Milne}) with $g$.
By using (\ref{Nf}),
\begin{equation}\label{id}
(g,{N}(g))=
\frac 1 2 (v_yg,g)-\frac 1 2 \mathfrak{U}(v_ xv_y g,g),
\nonumber\end{equation}
and so we obtain:
\begin{multline}\label{green}\frac{1}{2}\frac{ \partial}{\partial Y}(v_y g, g)+\frac 1 2 G\omega(v_y g, g)=
(w,Lw)+\frac{1}{2}G\mathfrak{U}\omega(v_xv_y g, g)\\=(w,Lw)+\frac{1}{2}G\mathfrak{U}\omega(v_xv_y w, w)+G\mathfrak{U}\omega(v_xv_y w, q),\end{multline}
because $(v_xv_y q,q)=0$.
The last but one term in the above Green identity is handled by adding in the cross section in the linearized Boltzmann operator  the unphysical term $\vartheta B(\,\cdot\, ,\,\cdot\,)^2$. Indeed, in this case the inequality (\ref{spectral}) holds with $\nu$ replaced by $\nu +\vartheta\tilde\nu$, with $\nu$ satisfying the inequalities (\ref{stimanu}) and $\tilde\nu$ such that
\begin{equation}\label{stimanutilde}\tilde\nu_0(1+|v|)^2\le \tilde\nu(v)\le \tilde\nu_1(1+|v|)^2,\quad v\in \mathbb{R}^3\end{equation}
for suitable constants $\tilde\nu_0$ and $\tilde\nu_1$. This allows to control the term $\frac{G\mathfrak{U}\o}2(v_xv_y w, w)$ as long as we have $\vartheta>\frac{G\mathfrak{U}}2$. \newline
So we have only to worry about the term ${\omega G\mathfrak{U}}(v_xv_yq,w)$ for which we use the bound
\[{\omega G\mathfrak{U}}|(v_xv_yq,w)|\le\frac 1 2 {\omega G\mathfrak{U}}\|w\|^2 + \frac c 2 {\omega G\mathfrak{U}}\|q\|^2.\]

We set $\mathcal{A}= (v_y g, g)$. Note that   $\mathcal{A}(0)<c$ by (\ref{123}). We need upper and lower bounds on $\mathcal{A}(0)$.
We can write
\begin{multline}\label{AL}\mathcal{A}(0)=\mathcal{A}(\ell)\exp\{\frac{G}{2}(\Omega(\ell)\}+\int_0^\ell dY \exp\{\frac{G}{2}(\Omega(\ell)-\Omega(Y)\}\\\Big[-(w,(L+\frac 1 2 \mathfrak{U} G\omega v_xv_y)w)(Y)-G\mathfrak{U}\omega(Y)(v_xv_y q,w)\Big].\end{multline}
By the reflecting boundary conditions, $\mathcal{A}(\ell)=0$. Moreover, $(w,(L+\frac 1 2 G \mathfrak{U}\omega v_xv_y)w)\ge-(w,\nu w)$ for $\vartheta>\frac 12{G\mathfrak{U}}$.
So we only have to estimate the last term. We use the following bound proved later:
\begin{equation}\label{boundq}\|q(Y)\|\le \sqrt{c+|\mathcal{A}(0)|}+ c\|\sqrt\nu w\|(Y)+\int_0^YdY'\|\sqrt\nu w\|(Y').\end{equation}
We use the bound
\begin{multline}\int_0^\ell\omega dY\left[\int_0^Y||w||(Y')dY' \right]^2\le \int_0^\ell\omega(Y) YdY\int_0^\ell||w||^2(Y')dY'\\\le c \int_0^Y||w||^2(Y')dY'.\end{multline}
Then we plug (\ref{boundq}) in (\ref{AL})  and use  the spectral inequality, to get
 the following  bound for $|\mathcal{A}(0)|$,  using the fact that $\omega$ is compactly supported,
 $$|\mathcal{A}(0)|\le G(c+c|\mathcal{A}(0)|),$$
 which implies the bound on $|\mathcal{A}(0)|$ for $G$ sufficiently small.
Using this one can conclude that
\begin{equation}\label{stimaw}\int_0^\ell dY \|w(Y)\|^2 <c\end{equation}
uniformly in $\ell$.

We need to prove the bound (\ref{boundq}). Let $\beta_\alpha=(v_y^2\bar\chi_\alpha q)$ for $\alpha\neq 2$. Since $\beta_\alpha=\sum_{\gamma\neq 2} A_{\alpha,\gamma}b_\gamma$, with $A_{\alpha,\gamma}=(v_y^2\bar\chi_\alpha,\bar\chi_\gamma)$ a positive non singular matrix, to bound $\beta_\alpha$ is equivalent to estimate $\|q\|$.
The equation for $\beta_\alpha$ is obtained by taking the inner product  of the first equation  in (\ref{Milne}) with $v_y\bar\chi_\alpha$.
The result is
\[\frac{\partial \beta_\alpha}{\partial Y}= G\omega\sum_{\gamma\neq2}\mathcal{B}_{\alpha,\gamma}\beta_\gamma+ \mathcal{D}_\alpha+ (v_y\bar\chi_\alpha, Lw),\]
with
\[\mathcal{D}_\alpha=-\frac{\partial}{\partial Y}(v_y^2\bar\chi_\alpha, w)-G\omega (v_y\bar\chi_\alpha, { N}w)\]
and
\[\mathcal{B}_{\alpha,\gamma}= G\omega\Big[
\delta_{\alpha,\gamma}(1+\mathfrak{U}(\delta_{\alpha,1}-U\delta_{\alpha,4})-(v_x^2\bar\chi_\alpha,\bar\chi_\gamma)+(\delta_{\alpha,1}-\mathfrak{U}\delta_{\alpha,4})(\bar v_x^2,v_y^2)A^{-1}_{\alpha,\gamma}\delta_{1,\gamma}
\Big].\]
An integration by parts shows that $|(v_y\bar\chi_\alpha, Nw)|\le c\|w\|$. The rest of the argument is as in \cite{CEM}. The only difference is in the estimate of $\beta_\alpha(0)$. We have
\[|(v_y^2\bar\chi_\alpha,g)(0)|\le (|v_y|g,g)^{1/2}(|v_y|^3,|\bar\chi_\alpha|).\]
\[(|v_y|g,g)(0)=\int_{v_y>0} v_y h^2-\int_{v_y<0} v_y g^2=2\int_{v_y>0} v_y h^2-\mathcal{A}(0).\]
By using (\ref{123}) we then get (\ref{boundq}).

To get  estimates uniform in $\ell$ also for $q$ we take the scalar product of the first equation in (\ref{Milne})  and ${L}^{-1}(\bar\chi_\a v_y)$. The term on the right hand side is $(L^{-1}(\bar\chi_\a v_y),Lw)$, which is zero by the orthogonality property. We get then an equation for $\Theta_\a=(v_y{L}^{-1}(\bar\chi_\a v_y),q)$ whose solution is
$$\Theta(Y)=\displaystyle{e^{-\int_0^YdsG\omega(s)(\mathcal{G}\mathcal{Q}^{-1})}}({L}^{-1}(v_y \bar\chi),v_y g)(0)-({L}^{-1}(v_y \bar\chi), v_yw)(Y)$$
$$+\int_0^Ydt\displaystyle{e^{-{}\int_t^YdsG\omega(s)(\mathcal{G}\mathcal{Q}^{-1})}}D(t)$$
where $\mathcal{G}$ and $\mathcal{Q}$  are suitable matrix, with $\Theta_\a=\mathcal{Q}_{\a\gamma}b_\gamma$,  $\mathcal{Q}$ invertible, and
$$D_\alpha(Y)=-G\omega(Y)\Big[({L}^{-1}(v_y \bar\chi_\alpha),{ N}w)-\frac{\partial}{\partial Y}({L}^{-1}(v_y \bar\chi_\alpha),v_y w)\Big].$$

 By the Schwartz inequality, boundedness of ${L}^{-1}$ and the  fact that $\omega$ has compact support, the last integral is finite.  Moreover, $ \parallel w\parallel$ vanishes at infinity and the first term on the right hand side   is finite. This implies that there exists a finite  limit at infinity, $\Theta^\infty$, of $\Theta$ and there is $Y_0$ such that for $Y> Y_0$ we have
$$|\Theta(Y)-\Theta^\infty|^2\le c||w||^2(Y) .$$
By the argument in \cite{CEM} we get also for $Y> Y_0$

$$|b(Y)-b^\infty|^2\le c||w||^2(Y) .$$

\vskip.2cm

\noindent The other arguments in \cite{CEM} can be adapted in a similar way.
Using above estimates, the exponential decay of $\|w\|$ and $|b_\alpha-b_\alpha^\infty|$ is established. The properties of the derivatives are also obtained with the method presented there, with a minor modification due to the special structure of the force in this case. Indeed, the $v_x$ and $v_y$ derivatives of $g$ satisfy in this case a coupled system of equations. Therefore, they are both controlled only away from the boundary.

 To state the final theorem, we define the norms
\begin{equation}\pa f\|_{q,2,\theta}= \left(\int_{\R^3}d\/vM(v)\left(\int_\theta^\infty d\/Y|f(Y,v)|^q \right)^{\frac 2 q}\right)^{\frac 1 2},\end{equation}
for $\theta\ge0$.
\begin{theorem}\label{milneth}
\item{1)} Suppose that for some $\beta>0$
$$\parallel e^{\beta Y} S\parallel _{2,2,0}<\infty,\quad \parallel e^{\beta Y} S\parallel _{\infty, 2,0}<\infty\ .$$
Then there is a unique solution $g\in L_2(\R^3, L_\infty(\R^+))\cap L_2(\R^+
\times \R^3)$   to the Milne problem  (\ref{Milne}).
Moreover there exist constants $c$ and $c'$ such that
$g$ verifies the conditions:
$$
\int_{\mathbb{R}^3}
M g dv=0,\quad g_\infty \in\hbox{\rm Null}\,\tilde L,
$$
$$\parallel e^{\beta' Y}(g(Y,v)-g_\infty (v))\parallel_{2,2}<c,\quad \parallel e^{\beta' Y}(g(Y,v)-g_\infty (v))\parallel_{\infty, 2}<c $$
for any $\beta' <c'$.

\item{2)} Suppose that for  $\ell\ge1$, $\theta>0$ and $i=1,\dots,3$
$$\int_{v_y>0} M
\Big|\frac{\partial^\ell h}{ \partial{v_i^\ell}}\Big|<\infty,\quad\Big\|e^{\beta Y}\frac{\partial ^\ell S}{ \partial{v_i ^\ell}}\Big\|_{2,2,\theta}<\infty, \quad\Big\|e^{\beta Y}\frac{\partial ^\ell S}{\partial{v_i ^\ell}}\Big\|_{\infty,2,\theta}<\infty.$$
Then there are
constants $c'$ and $c_\ell$ such that
\[\Big\|e^{\beta' Y}\frac{\partial ^\ell g}{\partial{v_i ^\ell}}\Big\|_{2,2,\theta}<c_\ell, \quad\Big\|e^{\beta' Y}\frac{\partial ^\ell g}{\partial{v_i ^\ell}}\Big\|_{\infty,2,\theta}<c_\ell.\]
for any $\beta'<c'$.
\end{theorem}

\section{The remainder}
\setcounter{equation}{0}
\setcounter{theorem}{0}
\setcounter{proposition}{0}
\setcounter{remark}{0}
\def\theequation{4.\arabic{equation}}
\def\thetheorem{4.\arabic{theorem}}
\def\theproposition{4.\arabic{proposition}}
\def\theremark{4.\arabic{remark}}

\subsection{Equation for the remainder}
It is immediate to check that the remainder $\mathcal{R}$ in (\ref{truexpansion}) has to satisfy the following equation:
\begin{multline}\label{remainder1}
\ v_y\frac{\partial \mathcal{R}}{\partial y} +\frac{\e^2}{\delta^2C^2} \sigma(y){v_x}\left(v_x\frac{\partial \mathcal{R}}{\partial v_y}-v_y\frac{\partial \mathcal{R}}{\partial v_x}\right)=\frac{1}{\e}\Big[\mathcal{L}_\delta\mathcal{R}+2Q(\Phi, \mathcal{R})\Big]+ Q(\mathcal{R},\mathcal{R}) +A,
\end{multline}
with the inhomogeneous term $A$ given by
\begin{multline}\label{A}
A=\sum_{h,k\ge1, h+k>N}\e^{h+k-2}Q(F_h,F_k)-\e^{N-1}v_y\partial_y B_N-\\-\frac{\e^2}{C^2\delta^2}\sum_{h=N-2}^{\infty}\sum_{n=0}^N\e^{h+n}\sigma^{(h)}\mathcal{N}(B_{n})-\frac{\e^2}{C^2\delta^2}\sum_{h=0}^{N-3}\e^{h}\sigma^{(h)}\sum_{n=N-2}^N\e^n\mathcal{N}(B_{n-h})\\
-\e^N\sum_\pm\left[\tilde{\mathcal{L}}^{\pm}_{\vartheta}   b_{N}^\pm+\frac{\e^{3}}{C^2\delta^2}\tilde\sigma^c(\mp(\e Y^\pm-\pi))\mathcal{N}(\bar b^\pm_{N,\infty})\right], \quad \vartheta=\frac{\e^{3}}{C^2\delta^2}
\end{multline}
 {In this section we need to assume $\delta=\gamma \epsilon^{\frac{2}{3}}$ with $\gamma$ sufficiently small.}\\
We recall that $M_\delta=M(1+\delta r,1+\delta \tau, (\delta U,0,0);v)$, where $( r, \tau, (U,0,0))$ is the solution of (\ref{HD}). We keep now the dependence on $\delta$ and denote by $M$ the standard Maxwellian.
We write $\mathcal{R}=M R$ and denote by $L$ the operator
\begin{equation}\label{defL} Lf= M^{-1}\mathcal{L}_0 Mf.\end{equation}
We use the Hilbert space $\mathbb{H}$ of the measurable functions on the velocity space $\mathbb{R}^3$ with inner product
\begin{equation}\label{inner1}(f,g)=\int_{\mathbb{R}^3}dv M(v)f(v)g(v).\end{equation}
Note that this involves $M$, while the one considered in Section 2,  (\ref{innerprod}) involved $M_\delta^{-1}$. The operator $L$ has the same properties already mentioned for $\mathcal{L}_\delta$ and we do not repeat them here. We just note that the null space of $L$ is given by the span of the $\chi_j$'s defined in (\ref{collinv}). We still denote the orthogonal projector on this kernel by $P$, and on its orthogonal complement by $P^\perp=1-P$.\\
Moreover, we write
\begin{equation}\label{Ldelta}M^{-1}\mathcal{L}_\delta Mf= Lf+ M^{-1}[\mathcal{L}_\delta -\mathcal{L}_0]Mf.\end{equation}
We also use the notation
\begin{equation}\label{J}J(f,g)=M^{-1}Q(Mf,Mg).\end{equation}
Note that
\begin{eqnarray*}\frac{M_\delta-M}{M}=&&\frac{1+\delta r}{(1+\delta \tau)^{3/2 }}\exp(-\frac{1}{2(1+\delta \tau)}[(v_x-\delta U)^2+{\hat v}^2]+\frac {v^2}{2})-1\\=&&
\frac{1}{2}\delta \tau v^2+\delta Uv_x +\delta r -\frac{3}{2} \delta \tau + \mathcal{O}(\delta^2).\end{eqnarray*}
\\
By (2.1) and (2.23)
$$\frac{\epsilon F_1}{M}=\epsilon P\frac {B_1}{M}+\mathcal{O}(\epsilon \delta).$$
 {We define}
$$W:=\frac{M_\delta-M}{M}+\frac{\epsilon F_1}{M}=\frac{1}{2}\delta \tau v^2+\delta Uv_x +\delta r -\frac{3}{2} \delta \tau +\epsilon P\frac {B_1}{M}+\mathcal{O}(\gamma^2\epsilon^{\frac{4}{3}})$$
 {Hence we have
$$W =W_1+\mathcal{O}(\gamma^2\epsilon^{\frac{4}{3}}),$$}
where $$W_1=\delta(\frac{1}{2} \tau v^2+ Uv_x + r -\frac{3}{2}  \tau) +\epsilon (\rho_1+\frac{v^2-3}{2}\tau_1).$$
In the present one dimensional context the $u$-part of $B_1$ is zero, so
\\
\begin{eqnarray}2J(W,Pf)&=& 2J(W_1,Pf)+2J(W-W_1,Pf)\notag
\\&=&-L(W_1Pf)+2J(W-W_1,Pf)\notag
\\&=&L[(-\delta Uv_x-\frac{1}{2}\delta\tau v^2-\frac{\e}{2}v^2\tau_1)Pf]+\mathcal{O}(\gamma^2\epsilon^{\frac{4}{3}})
\\&:=&L({a}Pf)+\mathcal{O}(\gamma^2\epsilon^{\frac{4}{3}}),\notag
\\a&=&-\delta Uv_x-\frac{1}{2}\delta\tau v^2-
\frac{\e}{2}v^2\tau_1\notag\end{eqnarray}

Here the property $2J(Pf,Pg)= -L(PfPg)$ has been used. The following operator will play a major role:
\begin{equation}\label{LJ}L_Jf:=Lf+L({a}Pf). \quad \end{equation}
But
$a=\mathcal{O}(\delta+\e)$, and hence the operator $L(a\cdot)$ is $\mathcal{O}(\delta+\epsilon)$.
Denote by $P_J$ the orthogonal projection on ${\rm Kern}\  L_J$, where  ${\rm Kern}\ L_J=
{\rm span} \{ \chi_j-L^{-1}L(a\chi_j), 0\leq j\leq 4\} $ as in \cite{AEMN2}, Section 2. It holds that $P_J=P-(I-P)aP$.

\begin{theorem} [Spectral gap property of $L_J$]
\label{2.2}
There is a constant $c$, such that for any function $f$ in $H$,
\begin{eqnarray} \label{2.7}
&& -((I+Pa)L_Jf,f)\geq c(\nu (I-P)((I+aP)f,(I-P)(I+aP)f).
\label{spgap}
\end{eqnarray}
\end{theorem}
\noindent\underline{Proof}.  First, $(I+Pa)L_J$ is self-adjoint for the scalar product $(f,g):=\int Mf(v)g(v)dv$. Indeed,
\begin{eqnarray*}
\int M((I+Pa)L_Jf(v)g(v)dv= \int ML((I+aP)f)(v)(I+aP)g(v)dv,
\end{eqnarray*}
and L is self-adjoint for $(.,.)$. The result follows  from the spectral inequality for $L$.
$\square$\\
\\
Constants which, independently of the parameter $\e$, can be made sufficiently small for the purposes of the proofs, will generically be denoted by $\eta$.\\
\\
Using these notation, equation (\ref{remainder1}) for the remainder becomes
\begin{eqnarray*}
\ v_y\frac{\partial {R}}{\partial y} +\frac{\e^2}{\delta^2C^2} \sigma(y){v_x}\left(v_x\frac{\partial {R}}{\partial v_y}-v_y\frac{\partial {R}}{\partial v_x}\right)\\
=\frac{1}{\epsilon}\Bigg[ LR+2J(W,R)+2J\left(\sum_2^5\e^j\frac{F_j}{M},R\right)\Bigg]+ J({R},{R}) +\mathcal{A},
\end{eqnarray*}
or
\begin{eqnarray}\label{remainder2}
\ v_y\frac{\partial {R}}{\partial y} +\frac{\e^2}{\delta^2C^2} \sigma(y){v_x}\left(v_x\frac{\partial {R}}{\partial v_y}-v_y\frac{\partial {R}}{\partial v_x}\right)=\frac{{L}_J{R}}{\epsilon} +\frac{H_1 R}{\epsilon}+ J({R},{R}) +\mathcal{A}.
\end{eqnarray}
Here $\mathcal{A}=M^{-1}A$ and
\begin{equation}
\label{H1} H_1f= 2J(W, (1-P)f)+2J(\sum_2^5\e^j\frac{F_j}{M},f)+2J(W-W_1,Pf).
\end{equation}
Note that, thanks to the arbitrariness left in the determination of the $F_N$, we can assume that
\begin{equation}
\label{Pa=0} P\mathcal{A}=0.
\end{equation}
The equation for the remainder has to be solved with the boundary conditions:
\begin{eqnarray}\label{bcR}& &R(-\pi,v)=\alpha_-(R)M^{-1}\tilde M_--\frac{1}{\e}\Psi(-\pi,v), \quad v_y>0\nonumber\\
\\
& & R(\pi,v)=\alpha_+(R)M^{-1}\tilde M_+-\frac{1}{\e}\Psi(\pi,v), \quad v_y<0,\nonumber
\end{eqnarray}
where
\begin{equation}M\Psi(\mp\pi,z,v,t)=\sum_{n=1}^n\e^n\Psi_{n,\e}(\mp\pi,z,v,t)\end{equation}
which are exponentially small in $\e^{-1}$ because of Theorem \ref{fj0},
and
\begin{eqnarray}\label{alphaR}
& & \alpha_-(R)= -\int_{v_y<0} dv v_y M[R(-\pi,v)+\frac{1}{\e}\Psi(-\pi,v)],\nonumber\\
\\
& & \alpha_+(R)= \int_{v_y>0} dv v_y M[R(\pi,v)+\frac{1}{\e}\Psi(\pi,v)].\nonumber
\end{eqnarray}

\subsection{Estimates for the remainder.}
We will proceed with the construction of the solution by  iteration,  based on estimates for a linearized problem where the non linear term $J(R,R)$ is computed at the previous step of the iteration. The generic term of the iteration will then satisfy a linear equation of the type
\begin{equation}\label{R}
v_y\frac{\partial R}{\partial y}+\frac{\e^2}{\delta^2C^2 }\sigma(y)v_x(v_x\frac{\partial R}{\partial v_y}-v_y\frac{\partial R}{\partial v_x})=\frac{1}{\e }[L_J R+ H_1(R)]+ g.
\end{equation}
At the  $n$-th step of the iteration, the term $g$ will be replaced by $\mathcal{A}+\e J(R^{n-1},R^{n-1})$. Therefore we assume
\begin{equation}Pg=0.\label{Pg=0}\end{equation}
The boundary conditions are
\begin{equation}\label{Rbb}R(\mp\pi,v)=\frac{\tilde{M}_\mp}{M}\int_{\mp w_y>0}(R(\mp\pi,w)+\frac{1}{\e}\Psi(\mp\pi,w))|w_y|Mdw\\-\frac{1}{\e}\Psi(\mp\pi,w)), \quad \pm v_y>0.
\end{equation}
\hspace{1cm}\\
In order to simplify the argument we also assume that the velocity of the inner cylinder vanishes, so that $M_-$ is the standard Maxwellian, up to the normalization. The argument can be easily adapted to the general case.\\
\begin{remark}\label{rem41} Note that the assumption (\ref{Pg=0}) and the diffuse reflection boundary conditions imply that any solution to (\ref{R}) is such that
\begin{equation}R_{v_y}:=\int_{\mathbb{R}^3} dv vy R M=0.\label{Rvy=0}
\end{equation}
Indeed, by integrating (\ref{R}) on velocities one gets:
\begin{equation}
\partial_y R_{v_y}+\frac{\e^2}{\delta^2C^2 }\sigma(y)R_{v_y}=0.
\end{equation}
By the boundary conditions then $R_{v_y}=0$.
\end{remark}
\hspace{1cm}\\
We first consider the linear problem with $g$ given and will get $L_2$ bounds.
This is done in two steps: first, we
get a control of the non-hydrodynamic part in terms of the hydrodynamic one. Second, we get an estimate of  the hydrodynamic part in terms of the non-hydrodynamic one, and then a bound for $R$ in terms of the known term $g$.\\
The norms used below are defined as follows:
\begin{equation}\label{norm}\pa f\|_{q,2}= \left(\int_{\R^3}d\/vM(v)\left(\int_{[-\pi,\pi]}d\/y|f(y,v)|^q\right)^{\frac 2 q}\right)^{\frac 1 2}.\end{equation}
Defining the ingoing velocity spaces   $\R^3_\pm$ at $y=\mp \pi$ as the sets $v=(v_x,v_y,v_z)$ such that $v_z\gtrless 0$
\begin{equation}\pa f\pa_{q,2,\sim}=\sup_\pm \left(\int_{\R^3_{\pm}} dv |v_y|M(v) \left(|f(\mp \pi, v)|^q\right)^{\frac 2 q}\right)^{\frac 1 2}.\end{equation}
We use the traces \begin{equation}\gamma^\pm f =\begin{cases} f|_{y=-\pi}, & \text{if } v_y\in \R^3_\pm ,\\f|_{y=\pi}, & \text{if } v_y\in \R^3_\mp .\end{cases}\end{equation}
Note that the norm $\parallel \,\cdot\,\parallel_{2,2,\sim}$ is defined only for incoming velocities. In the sequel, with an abuse of notation we will denote by $\parallel \gamma^- f\parallel_{2,2,\sim}$  the $\parallel \,\cdot\,\parallel_{2,2,\sim}$-norm of $S\gamma^- f$, where $S$ is the reflection of the $y$ component of the velocity.\\
\hspace{1cm}\\
We now establish the lemmas which allow to bound the solution to (\ref{R}) in these norms.\\
\begin{itemize}
\item{\bf{ Step  $1$}}
\end{itemize}
 {Recall} the notation $\zeta_j(v)=(1+|v|)^j$, $j\in\mathbb{N}$.
\begin{lemma}
Assume $\delta=\gamma\e^{\frac{2}{3}}$, and that $\parallel\zeta_{\frac{3}{2}}g\parallel^2_{2,2}<\infty$ . Then, for $\gamma$ small enough and for any $\eta>0$, the solution of (\ref{R}) satisfies \begin{eqnarray}\label{GI2tilde}
\frac{1}{\e  }\parallel \nu ^{\frac{1}{2}}(I-P)(1+aP){ R} \parallel _{2,2}^2
\leq c\Big[\e\eta\parallel P R\parallel _{2,2}^2 \\
+\epsilon^{\frac{1}{3}}\parallel \zeta_{\frac{3}{4}}(I-P) {g}\parallel^2_{2,2}+\frac{1}{\e^3}\parallel (1+|v|)^2\Psi\parallel_{2,2,\sim}^2  +\e^{\frac{7}{3}}\parallel\zeta_{\frac{3}{2}}(I-P) g\parallel_{2,2}^2\Big].\nonumber
\end{eqnarray}
\end{lemma}
\begin{proof}
{Define}
\begin{equation}
\label{key}
k(y)=\exp\left (\int_{-\pi}^y\frac{\e^2}{\delta^2C^2} \sigma(q)dq\right).
\end{equation}
Multiply (\ref{R}) by $2RM(1+a)k(y)$  and integrate over $v$. We get
\begin{eqnarray*}&&
\partial_y(v_y R,k(1+a)R))= {2}k[\e^{-1}((1+a)R,L_JR)
+\e^{-1}((1+a)R,H_1(R))\\&&+(1+a)R,g)]
+\frac{\delta}{\gamma^3 C^2}k\sigma\int dv R^2(\delta+\e)Uv_xv_yM+
k\int dvMv_y(\partial_y a) R^2 ,\end{eqnarray*}
or, with $\tilde R= R\sqrt{k(y)}$, $\tilde g= g\sqrt{k(y)}$ and  $\tilde \Psi= \Psi\sqrt{k(y)}$,\begin{eqnarray*}&&\partial_y(v_y(1+ a)\tilde R,\tilde R)= 2\big[\e^{-1}((1+a)\tilde R,L_J\tilde R)+\e^{-1}(\tilde R(1+ a),H_1(\tilde R)\\&&
+(\tilde R(1+ a),\tilde g)\big]+\frac{\delta}{\gamma^3 C^2}\sigma\int dv \tilde{R}^2(\delta+\e)Uv_xv_yM+ \int dvMv_y(\partial_y a){ \tilde R}^2 .\end{eqnarray*}
{We note that, by Remark \ref{rem41},  the last term vanishes, when $\tilde{R}$ is replaced by $P\tilde{R}$. Indeed, recall that $a=-(\delta Uv_x+\frac{1}{2}\delta\tau v^2+\e v_xU+
\frac{\e}{2}v^2\tau_1)$. Therefore, $v_y\partial_y a$ is odd in $v_y$.  {Since $R_{v_y}$ vanishes by Remark \ref{rem41} we conclude that  $$\int dv Mv_y(\partial_y a) (P\tilde R)^2=0.$$}}
Next, we have to look at the last term when $\tilde R^2$ is replaced by $2P\tilde R(1-P)\tilde R$. This is bounded by
$$C\delta\|P\tilde R\|_2^2\|\nu^{1/2}(1-P)\tilde R\|_2^2\le C(\delta^2 \|P\tilde R\|_2^2+\|\nu^{1/2}(1-P)\tilde R\|_2^2).$$
The remaining part of the last term when  $\tilde R^2$ is replaced by $((1-P)\tilde R)^2$ can be estimated as $H_1$ below. The preceding force term is of order $\frac{\delta^2}{\gamma^3}$, and can be estimated in the same way.\\
We next discuss the bounds for $((1+a)\tilde{R},L_J\tilde R)+((1+a)\tilde R,H_1(\tilde R))$. First\\
$$((1+aP)\tilde R,L_J \tilde R)\le  -c \parallel \nu ^{\frac{1}{2}}(I-P)(1+aP){\tilde R}\parallel_{2}^2,$$
$$|(a(1-P)\tilde R,L_J \tilde R)|\le  C
(\eta \parallel \nu ^{\frac{1}{2}}(I-P)(1+aP){\tilde R}\parallel_{2}^2+
\frac{\delta^2}{\eta}\|\nu^{\frac{1}{2}}(1+|v|^2)(I-P)\tilde{R}\|_2^2).$$
\begin{eqnarray*}
|((1+a)\tilde{R},H_1(\tilde R))|\le C(\delta^3\parallel P\tilde{R}\parallel_2^2+
\delta \parallel \nu^{\frac{1}{2}}(I-P){\tilde R}\parallel_{2}^2+\\
\delta^2\int dv M {|v|}^{4}|(1-P)\tilde R|^2).
\end{eqnarray*}
Here in e.g. the last term, after integrating with respect to $y$, the order of $|v|$ can be reduced by
writing $\nu =-L+K$, and using a standard estimate for $K$ (see \cite{Ma}),
\begin{eqnarray*}
\int (\zeta_{j+\frac{1}{2}}Kf)^2Mdv \leq C\int (\zeta_j f)^2Mdv.
\end{eqnarray*}
Multiply (\ref{R}) by $2RM\zeta_jk(y)$  and integrate over $v$. We get
\begin{eqnarray*}
\frac{\partial }{\partial y}(v_y R,k\zeta_jR))= {2}k[\e^{-1}(\zeta_j R,LR)
+\e^{-1}(\zeta_jR,2J(W,R))+(\zeta_jR,g)].
\end{eqnarray*}
It follows that
\begin{eqnarray*}
\epsilon\int k|v_y|\zeta_j\gamma^-(R)^2Mdv+ \int k\nu\zeta_jR^2Mdvdy =
\int k\zeta_jRK(R)Mdvdy\\
+\epsilon\int
k|v_y|\zeta_j\gamma^+(R)^2Mdv+ 2  \int k\zeta_jRJ(W, R)Mdvdy
+\epsilon \int k\zeta_j gRMdvdy.
\end{eqnarray*}
Hence
\begin{eqnarray}\label{GI0}
 \int \zeta_{j+1}R^2Mdvdy\leq c\Big((\epsilon+\delta) \int \zeta_{j+1}R^2Mdvdy
 +\int \zeta_{j-\frac{1}{2}}R^2Mdvdy\\
  +\epsilon\int k |v_y|\zeta_j\gamma^+(R)^2Mdv+\epsilon\int \zeta_{j-1}g^2Mdvdy\Big).\nonumber
\end{eqnarray}
In the present case $j=  {4}$, and we use (\ref{GI0}) three times and the assumption $\parallel\zeta_{\frac{3}{2}}g\parallel_{2,2}$ is finite. \\
Moreover, $$(\tilde R (1+a),\tilde g)=((I-P)((1+a)\tilde R ),(I-P)\tilde g)\le c\Big[\frac{\eta}{\e}\parallel \nu^{\frac{1}{2}}(I-P)((1+aP)\tilde R )\parallel_{2}^2 $$
$$+\gamma \e\parallel \zeta_{\frac{5}{4}}(I-P)\tilde{ R} \parallel_{2}^2 +{\e}^{\frac{1}{3}} \parallel \zeta_{\frac{3}{4}}(I-P)\tilde g\parallel_{2}^2\Big].$$ Putting all the estimates together,

\begin{eqnarray}\label{GI}
-{\mathcal B}+\frac{c_1}{\e  }\parallel \nu ^{\frac{1}{2}}(I-P)(1+aP){\tilde R} \parallel _{2,2}^2
\leq c\Big[\e^{\frac{1}{3}} \parallel \zeta_{\frac{3}{4}}(I-P)\tilde g\parallel_{2,2}^2
\nonumber\\
+\epsilon^{\frac{4}{3}}\int|v_y|\zeta_4\gamma^+(R)^2Mdv
+\epsilon^{\frac{7}{3}}\parallel \zeta_{\frac{3}{2}}
(I-P)\tilde{g}\parallel^2_{2,2}
+\e\eta \parallel P_J \tilde R\parallel_{2,2}^2\Big],\end{eqnarray}
where
${\mathcal B}:=(v_y(1+a),\tilde R^2(-\pi,v))-(v_y(1+a),\tilde R^2(\pi,v))$.\\
\\
\\
Following the argument in  \cite{EML}, one can show that the first term in ${\mathcal B}$ is bounded by
\begin{eqnarray*}
c\big(\frac{\eta}{\e}\parallel \nu^{\frac{1}{2}}(I-P)(1+aP)\tilde R\parallel^2_{2,2}+\eta \e \parallel P\tilde R \parallel^2_{2,2}+\e^{-3}\parallel\Psi\parallel^2_{2,2,\sim}\big).\quad
\end{eqnarray*}
For the second one, we shall first
simplify by changing the boundary conditions at $\pi$ from diffusive to given ingoing data. The following lemma states the equivalence between the two problems.
\begin{lemma}
Equation (\ref{R})  with the new boundary conditions
\begin{eqnarray}\label{Rbbnew}
R(-\pi,v)&&=\frac{\tilde{M}_-}{M}\int_{- w_y>0}(R(-\pi,w)+\frac{1}{\e}\Psi(-\pi,w))|w_y|Mdw\nonumber\\
&&-\frac{1}{\e}\Psi(-\pi,w), \quad  v_y>0,\\
R(\pi,v)&&=\frac{\tilde{M}_+}{M}\int_{ w_y>0}\frac{1}{\e}\Psi(\pi,w)|w_y|Mdw -\frac{1}{\e}\Psi(\pi,v),\quad  v_y<0, \nonumber
\end{eqnarray}
has the same solution as  problem (\ref{R}), (\ref{Rbb})

\end{lemma}

\begin{proof} We start by noticing that existence and uniqueness for the new  problem are classical.
The main  point is  that the new boundary condition implies
$$ \int_{ v_y>0}R(\pi,v) v_yMdv=0. $$
 In fact, since $\int g M dv=0$, by multiplication of (\ref{R}) by $k(y)M$ and integration over $v\in \R^3$    we get
\begin{equation}\label{K}
\frac{\partial }{\partial y}\left (k(y)\int v_y MR dv\right)=0,
\end{equation}
which implies
$$ k(\pi)\int v_y MR(\pi,v) dv=k(-\pi)\int v_y MR(-\pi,v) dv.$$
Hence, since $\int_{\R^3} v_y MR(-\pi,v)dv=0$, we have also $\int v_y MR(\pi,v) dv=0$.
 We write the  integral in  the left hand side  by using the boundary condition (\ref{Rbbnew})
\begin{eqnarray}0&&=\int_{\R^3} v_y MR(\pi,v)dv=\int_{v_y<0}v_yM_+dv\int_{ w_y>0}\frac{1}{\e}\Psi(\pi,w))w _y Mdw\\
&&+\int_{v_y>0}v_yMR(\pi,v)dv-\int_{v_y<0}v_y\frac{1}{\e}\Psi(\pi,v)
=\int_{v_y>0}v_yMR(\pi,v)dv,
\nonumber
\end{eqnarray}
because $\int_{v_y<0}v_yM_+dv=-1$ and  $\int_{\R^3} v_y\Psi(\pi,v)dv=0$. \newline Hence, $\int_{ v_y>0}R(\pi,v) v_yMdv=0$, and this implies that the unique solution $R$
of  the new problem also satisfies the old boundary conditions (\ref{Rbb}). \end{proof}

\begin{remark} Note that the condition $\int_{ \mathbb{R}^3}R(\pi,v) v_yMdv=0$ is crucial to make this argument work. However, while the old boundary conditions are constructed in such a way that this is true, the new boundary conditions do not ensure that it is automatically satisfied. It is the conservation law (\ref{K}) which ensures the vanishing of the outgoing flux also with the new boundary conditions.
\end{remark}
\hspace{1cm}\\
We now return to the estimate of the second term in $\mathcal B$. Writing\\
\\
$(1+a)(\pi)=1-\delta U_+v_x=(1-\frac{1}{2}\delta U_+v_x)^2-\frac{1}{4}\delta^2U_+^2v_x^2$,\\
\\
it follows that the outgoing part of $-(v_ya(\pi),\tilde R^2(\pi,v))$ is bounded from above by $\mathcal{O}(\delta^2)$. The ingoing part is bounded from above by\\
\\
$\frac{1}{\epsilon^2}\int Mdv (1+\delta U_+|v_x|) |v_y|\tilde\Psi^2(\pi,v).$\\
\\
Replacing in (\ref{GI}) we get

\begin{eqnarray}\label{GI1}
\frac{1}{\e  }\parallel \nu ^{\frac{1}{2}}(I-P)(1+aP){\tilde R} \parallel _{2,2}^2
\leq c\Big[
\e^{\frac{7}{3}} \parallel \zeta_{\frac{3}{2}}(I-P)\tilde g\parallel_{2,2}^2\quad \\
+\eta\e \parallel P \tilde R\parallel_{2,2}^2+\delta^2U_+^2\int_{v_y>0} Mv_x^2v_y\tilde R^2(\pi,v) dv
+\epsilon^{\frac{4}{3}}\int|v_y|\zeta_4\gamma^+(R)^2Mdv\nonumber\\+
\epsilon^{\frac{1}{3}}\parallel \zeta_{\frac{3}{4}}(I-P)\tilde{g}\parallel^2_{2,2}+\frac{c}{\e^3}\parallel\tilde \Psi\parallel^2_{2,2,\sim}\Big]
+\frac{1}{\e^2}\int (1+\delta U_+|v_x|)|v_y|\tilde\Psi^2(\pi,v) Mdv .\nonumber
\end{eqnarray}
Define
\begin{equation}
\label{key1}
k_1(y)=\exp\left (3\int_{-\pi}^y\frac{\e^2}{\delta^2C^2} \sigma(q)dq\right),
\end{equation}
and set $\bar{v}_x:= max\{|v_x|,1\}$. To estimate the outgoing integral of $\tilde R$ in the right hand side, multiply (\ref{R}) by $M R{k_1}v_x^2$, integrate in velocity over the region $v_y\ge q$, then over space using a smooth cut-off function $\chi(y)$ which is zero close to $y=-\pi$, and equal to one close to $\pi$, and finally over $q$ for $q_0\le q\le 0$ for $|q_0|$ small enough. Since we do not integrate over all $v\in\R^3$, we cannot use the spectral inequality to control the terms involving $L_J$, but will use only the boundedness of $L_J$. Notice that norms involving $\sqrt{k} R$ or $\sqrt{k_1} R$ are equivalent to the ones for $R$ since $k(y)$, $k_1(y)$ are functions uniformly bounded  in $\e$. The notation $\tilde{R}$ will be used for both. We obtain
\begin{eqnarray}\label{AAA}
&&\int _{q_0}^0dq\int_{v_y\ge q}dvM \bar{v}_x^2v_y \tilde R^2(\pi,v)\le c  |q_0|\Big[\frac{\delta}{\e}\parallel P \tilde R\parallel_{2,2}^2\nonumber \\
&&+\frac{\eta}{\delta^2\e}\parallel (I-P)(1+aP)\tilde R\parallel_{2,2}^2 +\frac{\e}{\delta}\parallel \nu^{-\frac{1}{2}}(I-P)\tilde g\parallel_{2,2}^2
\\&&+\frac{\delta}{\e}
\parallel\nu^{\frac{1}{2}}(1+v^2)(I-P)\tilde{R}\parallel_{2,2}^2\Big]
+\int_{q_0}^0dq\int_{v_y\ge q}dvdyM \chi'(y)\bar{v}_x^2  v_y \tilde R^2(y,v). \nonumber
\end{eqnarray}
The term on the l.h.s. equals
$$|q_0| \parallel \bar{v}_x\gamma^- \tilde R(\pi)\parallel_{2,2\sim}^2+\int_{q_0}^0 dq \int_{q\le v_y<0} dv \bar{v}_x^2v_y M \tilde R^2(\pi,v),$$
where
$$\gamma^- \tilde R(\pi)=  \tilde R(\pi,v),\quad v_y >0.$$
We move the second term in the previous expression to the r.h.s. of (\ref{AAA}) and bound it, by replacing  $\tilde R$ by the ingoing boundary data, as
\begin{multline}
\left|\int_{q_0}^0 dq \int_{0\le v_y<q} \frac{1}{M}\bar{v}_x^2v_y dv\left[ M_+\int_{w_y\ge 0} dw  w_y M \big(\tilde R(\pi,w)+\frac{1}{\e}\Psi(\pi,w)\big)\right. \right.\\ \left.\left.-M\frac{1}{\e}\Psi(\pi,v)\right]^2\right|
\le c(q_0) \big[\parallel\gamma^- \tilde R(\pi)\parallel_{2,2\sim}^2+\frac{1}{\e^2}\parallel (1+|v|)\tilde\Psi\parallel_{2,2,\sim}^2\big],\end{multline}
with $c(q_0)=o(|q_0|)$.
We replace this estimate in (\ref{AAA}) and divide both sides by $|q_0|$,
\begin{eqnarray}\label{pipi}
&&\parallel v_x\gamma^- \tilde R(\pi)\parallel_{2,2\sim}^2\le c  |\Big[\frac{\delta}{\e}\parallel P \tilde R\parallel_{2,2}^2\nonumber \\
&&+\frac{\eta}{\delta^2\e}\parallel (I-P)(1+aP)\tilde R\parallel_{2,2}^2 +\frac{\e}{\delta}\parallel \nu^{-\frac{1}{2}}(I-P)\tilde g\parallel_{2,2}^2
\\&&+\frac{\delta}{\e}
\parallel\nu^{\frac{1}{2}}(1+v^2)(I-P)\tilde{R}\parallel_{2,2}^2\Big]
+\frac{1}{\e^2}\parallel(1+|v|) \tilde\Psi\parallel_{2,2,\sim}^2. \nonumber
\end{eqnarray}
It is easy to bound also $\gamma^- \tilde R(-\pi)$, and by similar arguments to obtain
\begin{eqnarray}\label{pi}
\parallel\gamma^- \tilde R(-\pi)\parallel_{2,2\sim}^2 \le c  [\eta \parallel P \tilde R\parallel_{2,2}^2\nonumber +\frac{1}{\e}\parallel (I-P)(1+aP)\tilde R\parallel_{2,2}\\ +\frac{1}{\eta}\parallel \nu^{-\frac{1}{2}}(I-P)\tilde g\parallel_{2,2}^2
+\frac{1}{\e^2}\parallel(1+|v|) \tilde\Psi\parallel_{2,2,\sim}^2\big].
\end{eqnarray}
This gives an estimate for the ingoing boundary term in (\ref{GI1}), when using the boundary condition in the form (\ref{Rbbnew}). We also use the bound (\ref{pipi}) for $\parallel\bar{v}_x\gamma^- \tilde R\parallel_{2,2\sim}^2 $ in (\ref{GI1}) to get,
\begin{eqnarray}\label{GI2}
\frac{1}{\e  }\parallel \nu ^{\frac{1}{2}}(I-P)(1+aP){\tilde R} \parallel _{2,2}^2
\leq [\e\gamma^3+\e\eta]\parallel P\tilde R\parallel _{2,2}^2
 \nonumber\\
+\epsilon^{\frac{1}{3}}\parallel \zeta_{\frac{3}{4}}(I-P)\tilde{g}\parallel^2_{2,2}+\frac{1}{\e^3}\parallel (1+|v|)^2 \tilde{ \Psi}\parallel_{2,2,\sim}^2  +\e^{\frac{7}{3}}\parallel (1+|v|)^{\frac{3}{2}}(I-P)\tilde g\parallel_{2,2}^2.\quad\quad
\end{eqnarray}
  \end{proof}
\bigskip

\begin{itemize}
 \item{\bf{ Step $2$}}
 \end{itemize}
 \[\]
We will provide an a priori estimate for the hydrodynamic part of $R$. The equation
 for the remainder is

\begin{equation}\label{R1}
v_y\frac{\partial R}{\partial y}+\frac{\e^2}{\delta^2C^2 }\sigma(y)v_x(v_x\frac{\partial R}{\partial v_y}-v_y\frac{\partial R}{\partial v_x})=\frac{1}{\e }[L_J R+ H_1(R)]+ g
\end{equation}
with the boundary conditions
\begin{eqnarray}\label{Rbbnew1}
R(-\pi,v)&&=\frac{M_-}{M}\int_{- w_y>0}(R(-\pi,w)+\frac{1}{\e}\Psi(-\pi,w))|w_y|Mdw\nonumber\\
&&-\frac{1}{\e}\Psi(-\pi,w), \quad  v_y>0,\\
R(\pi,v)&&=\frac{M_+}{M}\int_{ w_y>0}\frac{1}{\e}\Psi(\pi,w)|w_y|Mdw -\frac{1}{\e}\Psi(\pi,v),\quad  v_y<0. \nonumber
\end{eqnarray}

\begin{lemma}\label{5.1}
If $Pg=0$, then the solution of (\ref{R1}), (\ref{Rbbnew1}) satisfies
\begin{equation}\label{PJ}
\parallel P R\parallel_{2,2}^2\le  c\big[\frac{1}{\e^2}\parallel \nu^{\frac 1 2}(I-P)(1+aP) R\parallel_{2,2}^2 +\parallel  g\parallel_{2,2}^2+\frac{1}{\e^2}\parallel \Psi\parallel_{2,2,\sim}^2\big].
\end{equation}
\end{lemma}
\begin{proof}
 Let  $\hat R=\mathcal{F}_y R$ be the Fourier-transform in $y$  of $R$ and $\xi_y\in \mathbb{Z}$ the conjugate variable to $y$. $\hat R$ satisfies the equation
\begin{equation}
iv_y\xi_y{\hat R}+\frac{\e^2}{\delta^2C^2 }\mathcal{F}_y\left\{\sigma(y)v_x(v_x\frac{\partial R}{\partial v_y}-v_y\frac{\partial R}{\partial v_x})\right\}=\e^{-1}(\widehat{L_J  R}+\widehat{ H_1( R)})-v_y r(-1)^{\xi_y}+\hat g
\label{FR}
\end{equation}
with
\begin{equation}
r(v)=
R(\pi,v)-R(-\pi,v).
\label{erre}
\end{equation}
With the notations $< R >$ for the $0$-Fourier term, and $\bar R := R - < R >$, we shall first give an estimate for $\bar R$. We use the notation
\begin{eqnarray*}
 &&\tilde{Z}= \e ^{-1} (\widehat{L_{J} R}+\widehat H_1(R))
\displaystyle{-\frac{\e^2}{\delta^2C^2 }\mathcal{F}_y\left\{\sigma(y)v_x(v_x\frac{\partial R}{\partial v_y}-v_y\frac{\partial R}{\partial v_x})\right\}}
\displaystyle{+\hat{g} -(-1)^{\xi_y}} v_y r,\\
&&\displaystyle{ Z= \e ^{-1} (\widehat{L_{J}R} +\widehat H_1(R))+\hat{g}-(-1)^{\xi_y}} v_y r, \\
&& \displaystyle{ Z'= \e ^{-1}( \widehat{L_{J}R}+\widehat H_1(R))+\hat{g}},\, \displaystyle{ \hat{U}= (i\xi_y v_y)^{-1}}.\end{eqnarray*}
Let $\mathfrak{h}$ be the indicatrix function of the set
$
\{ v\in \R^3; |v_y| <\alpha  \}
$,
for some positive $\alpha $ to be chosen later. For $\xi_y \neq 0 $,
\begin{eqnarray*}
\parallel P(\mathfrak{h} \hat{R }(\xi_y,\, \cdot\,))\parallel  &&\leq c \sum_{j=0}^4\left|\int_{\R^3}dv \mathfrak{h}(v)\hat{R }(\xi_y,v)M(v)\chi_j(v)\right|\\
&&\leq c\parallel \zeta_{-s} \mathfrak{h}\hat{R } \parallel  \sum_{j=0}^4 \parallel \mathfrak{h} \zeta_s{\chi}_j\parallel   \leq c\sqrt{{\alpha }}\parallel \zeta _{-s}\mathfrak{h} \hat{R } \parallel.
\end{eqnarray*}
We use this estimate with the following choice of $\alpha$,
$
\alpha = \parallel \zeta_{-s}\hat{R } \parallel_2 ^{-1}\parallel \zeta_{-s}Z'\parallel $.
We also introduce  an indicatrix function $\mathfrak{h}_1$ with $\alpha={\delta_1}$. We fix $\delta_1 $ so that $c\delta_1 << 1$. Then we find from the above estimate that the
$P$-part of the right-hand side, $\parallel~\hskip -.1cmP( \mathfrak{h}_1\hat{R } )\parallel $, can be  absorbed by  $\parallel P(\mathfrak{h}_1 \hat{R }) \parallel $ in the left-hand side. The estimates hold in the same way when $\mathfrak{h}_1$ is suitably smoothed around  $\sqrt{\delta_1} |\xi|$. For the remaining $\mathfrak{h}^c\mathfrak{h}_1^c\hat{R }=(1-\mathfrak{h})(1-\mathfrak{h})\hat{R }$, we shall use that
$\hat{R} = -\hat{U}\tilde{Z}$. Then
\begin{eqnarray*}
&&\parallel P(\mathfrak{h}^c_1\mathfrak{h}^c\hat{R }(\xi_y,\, \cdot\,) )\parallel ^2 \\&&\leq c\parallel \zeta _{s+2}\mathfrak{h}^c_1\mathfrak{h}^c\hat{U}\parallel ^2
\parallel \zeta _{-s}Z'\parallel ^2 +\frac{\parallel P(\mathfrak{h}^c_1\mathfrak{h}^c v_{y} r)\parallel ^2}{\delta_1|\xi_y|^2}
+ \frac{\e^2}{\delta^2C^2 }  |\Theta| \\
&&\leq \frac{c}{\mid \xi_y \mid ^2 \mid \alpha \mid }
\parallel \zeta _{-s}Z'\parallel ^2
+\frac{\parallel P(\mathfrak{h}^c_1\mathfrak{h}^c v_{y} r)\parallel ^2}{\delta_1|\xi_y|^2}+ \frac{\e^2}{\delta^2C^2 } |\Theta|,
\end{eqnarray*}
with
\begin{align*}\Theta=&\sum_{j=0}^4\int {\chi }_j\mathfrak{h}^c_1\mathfrak{h}^c \hat{U}\mathcal{F}_y\left\{\sigma(y)v_x(v_x\frac{\partial R}{\partial v_y}-v_y\frac{\partial R}{\partial v_x})\right\}\\
&\hspace*{150pt}Mdv \left( \int {\chi} _j\mathfrak{h}^c_1 \mathfrak{h}^c({\hat R}  -\hat{U} Z)M dv\right)^*.\end{align*}
We replace $\alpha $ by
 $\parallel \zeta_{-s}\hat{R } \parallel  ^{-1}\parallel \zeta _{-s}Z'\parallel  $ in the denominator.  That gives
\begin{eqnarray*}
 \parallel P\hat{R }(\xi_y,\, \cdot\,) \parallel ^2 && \leq c(\parallel \zeta _{-s}\hat{R } \parallel  \parallel \zeta _{-s}Z'\parallel   +\frac{\parallel P(\mathfrak{h}^c_1\mathfrak{h}^c v_{y} r)\parallel ^2}{\delta_1|\xi_y|^2}\\&&+\delta_1\parallel \zeta_{-s}(I-P)\hat{R } \parallel  ^2)+ \frac{\e^2}{\delta^2C^2 }|\Theta|.
\end{eqnarray*}
Hence,
\begin{eqnarray*}
&&\parallel P\hat{R } (\xi_y,\, \cdot\,)\parallel ^2  \leq c\Big( (\parallel P\hat{R } \parallel  +\parallel \zeta _{-s}(I-P)\hat{R } \parallel   )_2\parallel \zeta _{-s}Z'\parallel \\ && +\frac{\parallel P(\mathfrak{h}^c_1\mathfrak{h}^c v_{y} r)\parallel ^2}{\delta_1|\xi_y|^2}+\delta_1\parallel \zeta_{-s}(I-P)\hat{R} \parallel ^2\Big )+ \frac{\e^2}{\delta^2C^2 }|\Theta|.
\end{eqnarray*}
Consequently,
\begin{eqnarray*}
 &&\parallel P\hat{R } (\xi_y,\, \cdot\,)\parallel ^2 \leq c\Big(  \parallel \zeta _{-s}Z'\parallel ^2 +\frac{\parallel P(\mathfrak{h}^c_1\mathfrak{h}^c v_{y} r)\parallel ^2}{\delta_1|\xi_y|^2}+\parallel \zeta _{-s}(I-P)\hat{R } \parallel\   \parallel \zeta _{-s}Z'\parallel   \\
&& +\parallel\zeta_{-s}(I-P)\hat{R }\parallel  ^2 \Big)+ \frac{\e^2}{\delta^2C^2 }|\Theta|.
\end{eqnarray*}
We next discuss the term $ \frac{\e^2}{\delta^2C^2 } |\Theta|$. The first integral can be bounded by an integral of a product of $M$, $1+|\xi_y |$, a polynomial in $v $, $\mid  \hat{R }\mid $ and $\hat U$ or $\hat U ^2$. Since, by our choice of $\delta$, we have $\frac{\e^2}{\delta^2C^2 }=\frac{\e^{\frac{2}{3}}}{\gamma^2C^2 }$, this integral is bounded by $\e^{\frac{2}{3}} c \parallel \hat{R } \parallel_2   $. And so,
\begin{eqnarray*}
\parallel P\hat{R } (\xi_y,\, \cdot\,)\parallel  ^2\leq c\Big( \parallel \zeta _{-s}Z'\parallel  ^2+\frac{\parallel P(\mathfrak{h}^c_1\mathfrak{h}^c v_{y} r)\parallel ^2}{\delta_1|\xi_y|^2}+\parallel (I-P)\hat{R } \parallel  ^2\Big) .\\
\end{eqnarray*}
Therefore for $\xi_y\neq 0$,
\begin{eqnarray}
\int |P\hat{R } |^2(\xi_y ,v)Mdv
\leq c\Big( \frac{1}{\e ^2} \parallel \zeta_{-s}(v)\widehat{{L_J}R }( \xi_y , \cdot )\parallel_2  ^2
+\frac{1}{\e ^2} \parallel \zeta_{-s}(v)\widehat{{H_1}R }( \xi_y , \cdot )\parallel_2  ^2\nonumber\\
+\parallel (I-P)\hat{R } (\xi_y ,\cdot )\parallel  ^2
+\frac{\parallel P(\mathfrak{h}^c_1\mathfrak{h}^c v_{y} r)\parallel ^2}{\delta_1|\xi_y|^2}+ \parallel
\nu^{-\frac{1}{2}}\hat{g} (\xi_y ,\cdot )\parallel  ^2 \Big) \hspace{2cm} \label{3.3}.
\end{eqnarray}
The error from evaluating $P$ instead of $P_J$ is of order $\delta  {+\e}$. We remind that  the zero Fourier mode of $\overline{P\hat{R } }$ is zero by definition. Hence,
taking $\e$ small enough and summing over all $ 0\neq\xi_y \in  \mathbb{Z}$, implies by Parseval that
\begin{eqnarray}\label{3.8}
\int (\overline {PR })^2(y,v)Mdvdy
\leq c\Big( \frac{1}{\e ^2}\int \nu ((I-P)(1+aP){R} )^2(y,v)Mdvdy \nonumber \\
+\frac{\parallel P(\mathfrak{h}^c_1\mathfrak{h}^c v_{y} r)\parallel ^2}{\delta_1}+\int \nu ^{-1}{g}^2(y, v)Mdvdy+\delta\parallel PR\parallel^2_{2,2}\Big).
\label{uno}
\end{eqnarray}
The $r$-term can be expressed from (\ref{FR}) at $\xi _y= 0$
\begin{eqnarray}\label{r-estimate}
&&v_y r(v)= \frac{1}{\e }\widehat{{L}_{J}{R}}(0,v)-\frac{\e^2}{\delta^2C^2 }\mathcal{F}_y\left\{\sigma(y)v_x(v_x\frac{\partial R}{\partial v_y}-v_y\frac{\partial R}{\partial v_x})\right\}\Big|_{\xi_y=0}\\&&  +\hat{g}(0,v)+ \frac{1}{\e}\widehat{{ H_1}(R)}(0,v).\nonumber
\end{eqnarray}
Inserting this into (\ref{3.8}) results in
\begin{eqnarray*}
\int (\overline{PR })^2(y,v)Mdvdy
&&\leq c\Big( \frac{1}{\e ^2}\int \nu ((I-P)(1+aP)R)^2(y,v)Mdvdy \nonumber \\
&&+\int \nu ^{-1}g^2(y,v)Mdvdy+\delta \parallel PR\parallel _{2,2}^2\Big) .
\end{eqnarray*}
We are left with the Fourier component $P\hat{R}(\xi_y)$ for $\xi_y=0$,
and have $$P\hat R(0,v)=\sum_{\ell=0}^4\chi_\ell(v)\int dv M\psi_\ell\int dy R(y,v).$$ We discuss each moment separately.
Given two functions  $h(v)$ and  $f(\,\cdot\,, v)$ we use the notation
$f_{h}(\,\cdot\,):=\int dvh(v) f(\cdot,v)$. In particular, for $h=\chi_j$, $j=0,\dots,4$, we also use the notation $f_j$.

\begin{itemize}
\item $v_y$-moment:
\end{itemize}
  {As already noticed in Remark 4.1, m}ultiply (\ref{R1}) by $k(y)M$
and integrate over velocity. Now integrate over $ [-\pi, y]$ and $v\in\R^3$. Since  $\int v_y M(v)R(\pm \pi,v)dv=0$,
we have from (\ref{R1})
 $\hat R_{v_y}(0)=\int v_y M(v)R(y,v)dvdy=0$.
\begin{itemize}
\item $v_x$-moment:
\end{itemize}
We estimate the moment $\hat R_{v_x v^2_y}(0)$ and then use that
\[ \hat R_{v_x}=\hat R_{v_x v^2_y}-\hat R^\perp_{v_x v^2_y},
\]
with $\hat R^\perp=(1-P)\hat R$. \newline
Indeed, $\int  M PR v_xv_y^2dv= \int M R v_xdv$,
 because $\int Mv^2_x v^2_ydv=1$.
Multiply  equation (\ref{R1})
 by $M v_x v_y$, integrate  over $ [y, \pi]$ and $v\in\R^3$ and then integrate over $y\in[-\pi,\pi]$. We get
\begin{eqnarray}\label{vx}
|\int_{-\pi}^\pi dy\int dv v_xv_y^2MR(y,v)dv|\leq |2\pi \int dv v_xv_y^2MR(\pi,v)dv|\\
 +c\left[\e^{\frac{1}{3}}{\gamma^2}\parallel PR\parallel_{2,2} +\frac{1}{\e}\parallel (I-P)(1+aP)R\parallel_{2,2} +\parallel g\parallel_{2,2}\right].
\nonumber
\end{eqnarray}
In the boundary term, the integral over $v_y<0$ is easy because the boundary condition in $\pi$ is given
in terms of $\Psi$. To control the outgoing ($v_y>0$) part, we multiply again  equation (\ref{R1})
 by $2\pi M v_x v_y$, integrate this time  over $ [-\pi, \pi]$ and $v\in\R^3, v_y>0$. We get
\begin{eqnarray}\label{vx1}
&&| \int_{v_y>0} dv v_xv_y^2MR(\pi,v)dv|\leq
| \int_{v_y>0} dv v_xv_y^2MR(-\pi,v)dv|\\ && +c\left[\e^{\frac{1}{3}}{\gamma^2}\parallel PR\parallel_{2,2} +\frac{1}{\e}\parallel (I-P)(1+aP)R\parallel_{2,2} +\parallel g\parallel_{2,2}\right].\nonumber
\end{eqnarray}
To control the second term in  l.h.s. of (\ref{vx1}), we use the boundary condition in $-\pi$.
Finally, we replace the bound of the outgoing part given by equation  (\ref{vx1}) in   (\ref{vx}). The result is
\begin{eqnarray}
|\hat R_{v_x v^2_y}(0)|^2=\left|\int_{-\pi}^\pi dy\int  v_xv_y^2MR(y,v)dv\right|^2\le c[\e^{\frac{2}{3}}{\gamma^4 }\parallel PR\parallel^2_{2,2} \\
+\frac{1}{\e^2}\parallel (I-P)(1+aP)R\parallel_{2,2}^2 +\parallel g\parallel_{2,2}^2+\frac{1}{\e^2}\parallel (1+|v|)^2\Psi\parallel_{2\sim}^2]\nonumber
\end{eqnarray}
and the same estimate holds for $|\hat R_{v_x }(0)|^2$.
\begin{itemize}
\item $v_z$-moment:
\end{itemize} we get the same estimate for the $v_z$-moment and the proof is analogous, the only difference being that we start by multiplying (\ref{R1}) by $Mv_yv_z$.

\begin{itemize}
\item $\chi_4$-moment:
\end{itemize}
We recall $\chi_4=\frac{|v|^2-3}{\sqrt 6}$ and we denote $\hat R_4(0)=\int_{-\pi}^\pi  \chi_4 M Rdvdy$. We notice that
 \begin{eqnarray}
\hat{R}_{v^2_y\bar{A}}(0)=\frac{1}{\sqrt{6}}\hat{R}_{4}(0) \int v^2_y v^2\bar{A}Mdv+\hat{R}^\perp_{ v^2_ y\bar{A}}(0),
\label{3.10}
\end{eqnarray}
where  $\bar{A}$ is the non-hydrodynamic solution to
\begin{eqnarray}\label{AB1}
{L}(v_y\bar{A})= v_y(v^2-5).
\end{eqnarray}
To  control $\hat{R}_{v^2_y\bar{A}}(0)$ we multiply (\ref{R1}) by $Mv_y\bar{A}$ and proceed as before; first, we consider the integral $\int_{-\pi}^\pi dy\int_y^\pi dy' \int_{\R^3} dv $, then we study
$2\pi\int_{-\pi}^\pi dy\int_{v_y>0} dv $ and take the difference. Now, the analogous of the second term
in  (\ref{vx1}) is
$$2\pi \int_{v_y>0} dv v^2_y\bar{A}MR(-\pi,v)dv.$$
We use the boundary condition in $-\pi$, and observe that $\int v^2_y\bar{A}Mdv=0$ by orthogonality. Since the integral is even in $v_y$, also $\int_{v_y>0} v^2_y\bar{A}Mdv=0$.
\begin{itemize}
\item $\chi_0$-moment:
\end{itemize} Since $\displaystyle{\int  dvv_y^2 RMdv = R_0+\frac{2 R_{4}}{\sqrt 6}+\int  dvv_y^2 R^\perp Mdv}$, and we already have estimated $\hat{R}_{4}(0)$, it is enough to estimate the moment $\hat R_{v_y^2}(0)$. To this end, multiply (\ref{R1}) by $v_yM$ and integrate over $[y,\pi]\times \R^3$. Since $v_y$ is in ${\rm Kern}L$, we do not get contribution from the $L$, $H_1$ and $g$ terms in the right hand side. Moreover, by integration by parts, there is no contribution depending on $PR$ in the force term.  We have

\begin{eqnarray*}
-\int dv v_y^2 MR(y,v)dv+ 2\pi \int_{v_y<0} dv v_y^2MR(\pi,v)dv+2\pi \int_{v_y>0} dv v_y^2MR(\pi,v)dv\\
=\frac{\delta}{C^2\gamma^3}\int dv MR(v_y^2-v_x^2).
\end{eqnarray*}
The second term depends on the  ingoing flow at $\pi$, which is given in terms of $\Psi$. To control the  third term, we will as before estimate it using another equation. Multiply (\ref{R1}) by $2\pi v_yM$ and integrate over $[-\pi,\pi]\times \{v\in\R^3: v_y>0\}$.
\begin{eqnarray*}
 &&|\int_{v_y>0} dv v_y^2MR(\pi,v)dv|\leq | \int_{v_y>0} dv v_y^2MR(-\pi,v)dv|\\
&&+ c\left(\frac{\delta}{\gamma^3}\parallel R\parallel_{2,2}+\parallel g\parallel_{2,2} +\frac{1}{\e}\parallel (I-P)(1+aP)R\parallel_{2,2} \right).
\end{eqnarray*}
We get
\begin{eqnarray*}
&&|\int dv v_y^2 MR(y,v)dv|\leq |2\pi \int_{v_y>0}  v_y^2MR(-\pi,v)dv|\\
&& + c\left(\frac{\delta}{\gamma^3}\parallel R\parallel_2+\parallel g\parallel_2 +\frac{1}{\e}\parallel (I-P)(1+aP)R\parallel_2 +\frac{1}{\e}\parallel \Psi\parallel_{2\sim}\right).
\end{eqnarray*}
To estimate the first term to the right, we cannot employ the previous symmetry arguments, but will instead
use Schwartz' inequality
$$\int_{v_y>0}  v_y^2M|R|(-\pi,v)dv\le \left[\int_{v_y>0} v_y^3M\right]^{1/2}\left[\int_{v_y>0}v_yMR^2(-\pi,v)dv\right]^{1/2}.$$
We can now use (\ref{pi}) to estimate the last integral, and get
\begin{eqnarray*}
|\hat R_{v_y^2}|^2\le c [ \eta \parallel P  R\parallel_{2,2}^2\nonumber +\frac{1}{\e^2}\parallel (I-P)(1+aP) R\parallel_{2,2}^2 +\parallel  g\parallel_{2,2}^2+\frac{1}{\e^2}\parallel \Psi\parallel_{2,2,\sim}^2\big].
\end{eqnarray*}
\bigskip
Collecting all the moment estimates for $\xi_y=0$, we get
\begin{equation}
\parallel P\hat R(0)\parallel^2_2\le c \big[ \eta \parallel P  R\parallel_{2,2}^2,\nonumber +\frac{1}{\e^2}\parallel (I-P)(1+aP) R\parallel_{2,2}^2 +\parallel  g\parallel_{2,2}^2+\frac{1}{\e^2}\parallel (1+|v|)^2\Psi\parallel_{2,2,\sim}^2\big]
\end{equation}
and this concludes the proof of Lemma \ref{5.1}.
\end{proof}
\bigskip
\hspace{1cm}\\
Combining the two steps we have proved
\begin{theorem}\label{AE}
Assume $\delta=\e^{\frac{2}{3}}\gamma$, $Pg=0$, and $\parallel (1+|v|)^{\frac{3}{2}}g)\parallel_{2,2}$ finite. Then, for $\gamma$ small enough,  the solution of (\ref{R1}), (\ref{Rbbnew1}) satisfies
\begin{eqnarray}\label{AE1}
\parallel P  R\parallel^2_{2,2}\le c\big( \epsilon^{\frac{4}{3}}\parallel\zeta_{\frac{3}{2}}g\parallel_{2,2}^2+\frac{1}{\e^4}\parallel \zeta_2\Psi\parallel_{2,2,\sim}^2+\e^{-\frac{2}{3}}\parallel \zeta_{\frac{3}{4}}g\parallel^2_{2,2}\big)\nonumber\\
\\
\parallel \nu ^{\frac{1}{2}}(I-P)(1+aP){ R} \parallel _{2,2}^2
\leq c
\big(\frac{1}{\e^2}\parallel\zeta_2 \Psi\parallel_{2,2,\sim}^2  +\e^{\frac{10}{3}}\parallel\zeta_{\frac{3}{2}}g\parallel_{2,2}^2
+\e^{\frac{4}{3}}\parallel\zeta_{\frac{3}{4}}g\parallel_{2,2}^2
\big).\nonumber
\end{eqnarray}

\end{theorem}

\begin{proof} Replace (\ref{GI2tilde}) in (\ref{PJ}), and use  $\gamma$ and $\eta$ small. \end{proof}
The step from $L_2$ to $L_\infty$ is done by studying the solution along the characteristics as \cite{AEMN1}. The result is

$$\parallel \nu^{\frac{1}{2}} {\zeta_2}R\parallel^2 _{\infty ,2}
\leq c\Big(
\frac{1}{\e^2 }\parallel  {\zeta_2}R \parallel^2 _{2,2}
+\e^2 \parallel \nu^{-\frac{1}{2}} {\zeta_2}g\parallel^2_{ \infty, 2}
+\frac{1}{\e^2  }\parallel {\zeta_2} \bar{\Psi}\parallel^2_{2,2,\sim}\Big) .
 $$
 Then, by Theorem \ref{AE},
\begin{multline}\label{infty}
\parallel \nu^{\frac{1}{2}} {\zeta_2}R\parallel^2 _{\infty ,2}
\leq c\Big( \e^{-\frac{2}{3} }\parallel \zeta_{\frac{3}{2}}g\parallel^2 _{2,2}
+\e^2 \parallel  {\zeta_{\frac 3 2}}g\parallel^2_{ \infty, 2}
\\+\e^{-\frac{8}{3}}\parallel \zeta_{\frac{3}{4}}g\parallel^2_{2,2}+
\frac{1}{\e^6}\parallel\zeta_2\bar{\Psi}\parallel^2_{2,2,\sim}\Big) .
\end{multline}

\bigskip

With the a priori estimates provided by Theorem \ref{AE}, we are now in the position of proving the existence theorem for the remainder equation (\ref{R}-\ref{Rbb}), by an iteration procedure.

\begin{theorem}
There exists an isolated  solution in $L_{2}( [-\pi,\pi] \times \R^3; Mdvdy)$ to the problem (\ref{remainder2}), (\ref{bcR}).
\end{theorem}

\begin{proof} The remainder term $R$ will be obtained as the limit of the
approximating sequence $R^n$, where $R^0 = 0$ and

\begin{multline}\label{Rn}
v_y\frac{\partial R^{n+1}}{\partial y}+\frac{\e^2}{\delta^2C^2 }\sigma(y)v_x(v_x\frac{\partial R^{n+1}}{\partial v_y}-v_y\frac{\partial R^{n+1}}{\partial v_x})=\frac{1}{\e }[L_J R^{n+1}+ H_1R^{n+1}]\\+ {J}(R^{n},R^n)+\mathcal{A},
\end{multline}
\begin{eqnarray}\label{Rbbn}
R^{n+1}(-\pi,v)&&=\frac{M_-}{M}\int_{- w_y>0}(R^{n+1}(-\pi,w)+\frac{1}{\e}\Psi(-\pi,w))|w_y|Mdw,\nonumber\\
&&-\frac{1}{\e}\Psi(-\pi,w), \quad  v_y>0,\\
R^{n+1}(\pi,v)&&=\frac{M_+}{M}\int_{ w_y>0}\frac{1}{\e}\Psi(\pi,w)|w_y|Mdw -\frac{1}{\e}\Psi(\pi,v),\quad  v_y<0 .\nonumber
\end{eqnarray}
Here $\mathcal{A}$ is of order  $\e^4$ with $(1+|v|)^{\frac{3}{2}}\mathcal{A}$ quadratically integrable  {by Theorem\ref{fj0}}, and it holds that $R^{n+1}$ satisfies also the boundary conditions (\ref{Rbb}). The function $R^1$ is solution to (\ref{R1}-\ref{Rbbnew1}) with $g=\mathcal{A}$.
Then, by using the  a priori estimates of Theorem \ref{AE}, (\ref{GI0}), and (\ref{infty}), together with the exponential decrease of ${\Psi}$, we obtain, for some constant $c_1$,
\begin{eqnarray*}
\parallel \zeta_2R^1\parallel _{\infty ,2}\leq c_1\e^{\frac{2}{3}},\quad \parallel \zeta_2R^1\parallel _{2,2}\leq c_1 \e^{\frac{5}{3}}.
\end{eqnarray*}
By induction  for $\e$ sufficiently small,
\begin{eqnarray*}
&&\parallel \zeta_2 R^n\parallel _{\infty ,2}\leq 2c_1 {\e}^{\frac{2}{3}} , \\
&&\parallel\zeta_2(R^{n+1}-R^n)\parallel _{2,2}\leq c_2{\e^{\frac{1}{3}}}\parallel\zeta_2( R^n-R^{n-1})\parallel _{2,2},\quad n\geq 1,
\end{eqnarray*}
for some constant $c_2$. Namely,
\begin{eqnarray*}
&&v_y\frac{\partial (R^{n+2}-R^{n+1})}{\partial y}+\frac{\e^2}{\delta^2C^2 }\sigma(y)v_x(v_x\frac{\partial (R^{n+2}-R^{n+1})}{\partial v_y}-v_y\frac{\partial (R^{n+2}-R^{n+1})}{\partial v_x}) \\
&&= \frac{1}{\e }L_J(R^{n+2}-R^{n+1})+\frac{1}{\e }H_1(R^{n+2}-R^{n+1})
+G^{n+1},\\
&&(R^{n+2}-R^{n+1})(-\pi,v)=\frac{M_\mp}{M}\int _{w_z\lessgtr 0} (R^{n+2}-R^{n+1})(-\pi,w)|w_y|M_-dw,
\, v_y> 0,
 \\&& (R^{n+2}-R^{n+1})(\pi,v)=0, \, v_y< 0.
\end{eqnarray*}
Here, $
G^{n+1}= (I-P)G^{n+1}= {J}(R^{n+1}+R^n,R^{n+1}-R^n).
$ It follows that
\begin{eqnarray*}
&&\parallel \zeta_2( R^{n+2}-R^{n+1})\parallel _{2,2}\leq {c} \e^{-\frac{1}{3}}\parallel \zeta_1G^{n+1}\parallel _{2,2}\\
&&\hskip1.5cm\leq {c}\e^{-\frac{1}{3}}\Big( \parallel \zeta_2R^{n+1}\parallel _{\infty ,2}+\parallel \zeta_2 R^{n}\parallel _{\infty ,2}\Big) \parallel \zeta_2( R^{n+1}-R^{n})\parallel _{2,2}\\
&&\hskip1.5cm\leq c_2\e^{\frac{1}{3}}\parallel\zeta_2( R^{n+1}-R^{n})\parallel _{2,2}.
\end{eqnarray*}
Consequently,
\begin{eqnarray*}
\parallel \zeta_2R^{n+2}\parallel _{2,2}&&\leq \parallel \zeta_2(R^{n+2}-R^{n+1})\parallel _{2,2}+...+\parallel \zeta_2(R^{2}-R^{1})\parallel _{2,2}\\&&
+\parallel \zeta_2R^1\parallel _{2,2}
\leq 2c_1\e ^{\frac{5}{3}} ,
\end{eqnarray*}
for $\e $ small enough. Similarly, $\parallel \zeta_2R^{n+2} \parallel _{\infty,2}\leq 2c_1 {\e}^{\frac{2}{3}}$. In particular $(R^n)$ is a Cauchy sequence in $L_M^{2}( [-\pi,\pi]^2 \times \R^3)$. The existence of a solution ${R}$ to (\ref{remainder2}) follows. Local uniqueness follows along the same path.
\end{proof}

\begin{corollary}\label {co2}
There exists an isolated, non-negative $L_2$-solution $F$ to (\ref{basic1}), (\ref{bc1})  such that
\[\parallel M^{-1}[F-M_\delta]\parallel_{2,2}\le c\gamma\e^{\frac{4}{3}}.\]
\vskip .5cm
\end{corollary}

\begin{proof}[Proof of Theorem \ref{main1}]
 The non-negativity can be proved similarly to \cite{AN1}--\cite{AN4}. Then,
by Corollary \ref{cor} and \ref{co2} we have, for $q=2,\infty$
\begin{eqnarray*}&&\parallel M^{-1}[F-M(1,1,(\delta U,0,0))]\parallel_{q,2}\le \parallel M^{-1}[F-M_\delta]\parallel_{q,2}\\
&&\qquad+\parallel M^{-1}[M_\delta-M(1,1,(\delta U,0,0))]\parallel_{q,2}\le c(\e\delta +\delta^2)\end{eqnarray*}
which implies (\ref{basic2}) by taking into account the relation $\delta=\gamma\e^{\frac{2}{3}}$.
\end{proof}

\medskip

\end{document}